\documentclass[sigconf,nonacm]{acmart}

\usepackage{mathrsfs}
\usepackage{amsmath,amsfonts,bm}
\usepackage{graphicx}
\usepackage{setspace}
\usepackage{textcomp}
\usepackage{xcolor}
\usepackage{graphicx}
\usepackage{commath}
\usepackage{amsthm}
\usepackage{algorithm}
\usepackage[noend]{algorithmic}
\usepackage{eqparbox,array}

\usepackage{array}
\usepackage{mdwmath}
\usepackage{mdwtab}
\usepackage{eqparbox}
\usepackage[tight,footnotesize]{subfigure}
\usepackage{stfloats}
\usepackage{tabularx}
\usepackage{multirow}

\usepackage{url}
\usepackage{booktabs}
\usepackage{dsfont}
\usepackage{tikz}
\usepackage{slashbox}
\usepackage{bbm}

\newtheorem{theorem}{Theorem}
\newtheorem{definition}{Definition}
\newtheorem{lemma}{Lemma}

\newtheorem{proposition}{Proposition}

\newtheorem{assumption}{Assumption}

\makeatletter
\newcommand*{\ov}[1]{%
	$\m@th\overline{\raisebox{0pt}[1.3\height]{#1}}$%
}

\makeatletter
\def\@copyright{\relax}
\makeatother
\allowdisplaybreaks[3]
\def\BibTeX{{\rm B\kern-.05em{\sc i\kern-.025em b}\kern-.08emT\kern-.1667em\lower.7ex\hbox{E}\kern-.125emX}}

\begin{document}

%
\title{Learning-NUM: Network Utility Maximization with Unknown Utility Functions and Queueing Delay}

%
 

%

\author{Xinzhe Fu}
\affiliation{%
	\institution{LIDS, Massachusetts Institute of Technology}
	\country{USA}}

\author{Eytan Modiano}
\affiliation{%
	\institution{LIDS, Massachusetts Institute of Technology}
	\country{USA}
}
%
\begin{abstract}
	Network Utility Maximization (NUM) studies the problems of allocating traffic rates to  network users in order to maximize the users' total utility subject to network resource constraints. In this paper, we propose a new NUM framework, Learning-NUM, where the users' utility functions are unknown apriori and the utility function values of the traffic rates can be observed only after the corresponding traffic is delivered to the destination, which means that the utility feedback experiences \textit{queueing delay}. 
	The goal is to design a policy that gradually learns the utility functions and makes rate allocation and network scheduling/routing decisions so as to maximize the total utility obtained over a finite time horizon $T$. In addition to unknown utility functions and stochastic constraints, a central challenge of our problem lies in the queueing delay of the observations, which may be unbounded and depends on the decisions of the policy.
	We first show that the expected total utility obtained by the best dynamic policy is upper bounded by the solution to a static optimization problem. Without the presence of feedback delay, we design an algorithm based on the ideas of gradient estimation and Max-Weight scheduling. To handle the feedback delay, we embed the algorithm in a parallel-instance paradigm to form a policy that achieves $\tilde{O}(T^{3/4})$-regret, i.e., the difference between the expected utility obtained by the best dynamic policy and our policy is in $\tilde{O}(T^{3/4})$. Finally, to demonstrate the practical applicability of the Learning-NUM framework, we apply it to three application scenarios including database query, job scheduling and video streaming. We further conduct simulations on the job scheduling application to evaluate the empirical performance of our policy.
\end{abstract}

\maketitle
\section{Introduction}
Network Utility Maximization (NUM) has been a central problem in networking research for decades and has become a standard framework for making intelligent network resource allocation decisions. It has found a wide range of applications such as congestion control in the Internet \cite{cite:NUM-Kelly, cite:NUM-Low, cite:NUM-Chiang}, power allocation in wireless networks \cite{cite:NUM-NeelyPower} and job scheduling in cloud computing \cite{cite:NUM-Job1,cite:NUM-Job2}. 

As a network optimization paradigm, NUM studies the problems of user traffic admission control to maximize the users' total utility subject to network resource constraints. Previous works in NUM can be classified into two categories: static and stochastic.  In the static approach \cite{cite:NUM-Kelly, cite:NUM-Low, cite:NUM-Chiang, cite:NUM-Multipath, cite:NUM-Asu}, the traffic rates are modeled as flow variables, the bandwidth constraints are modeled as network flow constraints, and the analysis focuses on the convergence rates of the optimization algorithms. In the stochastic approach \cite{cite:NUM-NeelyFairness,cite:NUM-NeelyPower,cite:NUM-Longbo, cite:NUM-Job1,cite:NUM-Job2}, the traffic rates are determined by the time-average admitted traffic, the resource constraints are captured by the long-term stability of the stochastic queueing networks and the analysis focuses on the tradeoff between the long-term average utility and queue length.

Regardless of the differences in modeling and analysis, previous NUM results rest on a key assumption that the utility functions of network users are known. This is justified when the utility functions are simply optimization proxies for network performance criteria such as fairness \cite{cite:NUM-Kelly, cite:NUM-Low, cite:NUM-NeelyFairness}. However, when the utility represents more concrete quantities such as power and energy consumption \cite{cite:NUM-NeelyPower,cite:NUM-Longbo}, user satisfaction \cite{cite:NUM-Video1,cite:NUM-Video2} and job quality \cite{cite:NUM-Job1, cite:NUM-Job2, cite:NUM-Job3}, often we do not have prior knowledge of the utility functions, i.e., the functional relationship between the traffic rate of a user and its corresponding utility value is unknown in advance. 


In this paper, we propose a new NUM framework, Learning-NUM (L-NUM), where the utility functions are unknown but their values can be learned over the process of decision making. Specifically,  we consider a time-varying stochastic queueing network in discrete time, which captures both wireline and wireless networks. There are $K$ users, where user $k$ has a concave utility function $f_k$ that is initially unknown to the network operator. Each user has a corresponding source-destination pair in the network. At every time $t$, for each user $k$, the network operator injects a ``job'' of size $r_k(t)$ from the user's source to be delivered to the user's destination. The job size in our framework resembles the admitted traffic rate in the traditional NUM formulation. We will explain the connection between the two notions in Section \ref{sec:model}. Next, the operator chooses a network action that controls the routing and scheduling, which further determines the queue dynamics of the network. Finally, the utility value ($f_k(r_k)$) of a job (of size $r_k$) can only be observed after the job gets delivered to the destination as feedback from the user.

We study the problem of designing a policy that jointly determines the job sizes and network actions based on the utility function values learned from observations.
We define the utility achieved by a policy as the total utility of the jobs delivered by a finite time horizon $T$. This definition naturally enforces network resource constraints as the undelivered jobs in the queues at time $T$ are not counted towards the utility. We seek to design a policy with regret sublinear to $T$, where
regret \cite{cite:NUM-Qingkai} is defined as the gap between the expected utility of the policy and that of the optimal policy that has full knowledge of the utility functions in advance.
As a first step, we establish that the expected utility achieved by the optimal (dynamic) policy is upper bounded by $T$ times the optimal value of a static optimization problem, whose objective is the sum of the (unknown) utility functions and the constraints are implicitly given by the capacity region of the network.
This result provides an important insight that a policy achieves low regret if it can closely track the solution to the static optimization problem. 

While solving an optimization problem with unknown objective function is a common challenge faced in the online convex optimization literature \cite{cite:oco,cite:convexbandit}, our problem is further complicated by the facts that
the constraints, which essentially enforce network stability, are stochastic and unknown in advance (See {Section \ref{sec:upperbound}} for details). Thus, they cannot be handled by techniques in online optimization that require the feasibility region to be known in advance \cite{cite:convexbandit}. Moreover, the utility value can be observed only after the delivery of the job, which, in a First-In-First-Out network, happens after the delivery of the jobs injected before it. This means that the feedback in our problem experiences \textbf{queueing-style delay} that may be unbounded and depends on the decisions of the policy. Such delay evades existing techniques in the literature as they typically assume bounded or decision-independent delay \cite{cite:delayfeedback1,cite:delayfeedback2}.

 To deal with unknown utility functions and stochastic constraints,  we combine the ideas of gradient sampling \cite{cite:flaxman} and max-weight scheduling (back-pressure routing) \cite{cite:backpressure} to propose an online scheduling algorithm that works for the L-NUM problem without feedback delay. We next
  embed the algorithm into a parallel-instance paradigm to obtain a scheduling policy that can handle the queueing-style feedback delay and 
  achieve $\tilde{O}(T^{3/4})$-regret\footnote{$\tilde{O}(\cdot)$ hides logarithmic factors of $T$.}.
Finally, we show how to apply our framework to applications including database query \cite{cite:Database}, job scheduling \cite{cite:MLtask} and video streaming \cite{cite:NUM-Video1,cite:NUM-Video2}. We further empirically evaluate the performance of the proposed policy through simulations on job scheduling scenarios.


The rest of the paper is organized as follows. The model and formal definitions of the L-NUM framework are presented in Section \ref{sec:model}. In Section \ref{sec:upperbound}, we prove the upper bound on the optimal expected utility. In Section \ref{sec:policy}, we propose the online scheduling algorithm and the parallel-instance paradigm for the L-NUM framework. We further illustrate several applications of L-NUM in Section \ref{sec:application}. The empirical performance of the online scheduling policy is evaluated in Section \ref{sec:simulations}.  Finally, we conclude the paper with some future directions in Section \ref{sec:conclusion}.




\section{Model and Problem Formulation} \label{sec:model}
In this section, we specify the general network model and set up the framework of network utility maximization with unknown utility functions.
We consider a network $\mathcal{G}(\mathcal{V},\mathcal{E})$ with $\mathcal{V}$ being the set of nodes and $\mathcal{E}$ being the set of directed links. For each node $i\in\mathcal{V}$, we will denote its set of outgoing neighbors by $\mathcal{N}_i$.  There are $K$ classes of users in the network. Each user $k$ corresponds to a job class (also denoted by $k$), and is mapped to one source-destination pair $(s_{k},d_{k})$. Multiple job classes can be mapped to the same source-destination pair. Source $s_{k}$ sends class-$k$ jobs that get delivered to $d_{k}$ through the network. We will refer to the jobs sent from $s_k$ destined to $d_k$ as class-$k$ traffic. Each node $i\in\mathcal{V}$ has a queue $Q_i^k$ that buffers the incoming class-$k$ traffic of node $i$. The network operates in discrete time $t=1,\ldots,T$, where $T$ is the specified time horizon. At each time $t$, the network is in state $\omega(t) \in\mathcal{W}$, with $\mathcal{W}$ denoting the set of possible network states. In concrete applications, the network states may correspond to channel states of links, service states of servers, or simply a placeholder when the network is static with only one state. 
We assume $\omega(t)$'s is a sequence of i.i.d. random element with $\mathbb{P}(\omega(t)=\omega)=p(\omega)$. However, the distribution of $\omega(t)$ is unknown to the network operator.

\subsection{Traffic Model and Network Dynamics}
At each time $t$, the network operator first observes the current network state $\omega(t)$. It next chooses job size $r_k(t)$ for each class and sends a job of size $r_k(t)$ to the buffer $Q^{k}_{s_{k}}$, where $r_k(t)$ is a real value that satisfies $0\le r_k(t)\le B$. The job size corresponds to the amount of admitted traffic at a time slot. For example, as we will demonstrate in Section \ref{sec:application}, in video streaming, the job size represents the resolution of a video chunk sent to the user. We adopt this discrete notion of job size rather than the continuous notion of traffic rate in the traditional NUM framework because job size is more suitable for our finite-horizon discrete-time framework.
Finally, the network operator chooses a network action $\bm{x}(t)\in\mathcal{X}$ that incorporates the routing and scheduling decisions of the network. The feasible set of actions $\mathcal{X}$ can be discrete or continuous. For each $\bm{x}\in\mathcal{X}$, under network state $\omega$, we use $A_{ij}^k(\omega,\bm{x})$ to denote the offered transmission rate on link $(i,j)$ for class $k$, i.e., the amount of class-$k$ traffic that can be sent from node $i$ to node $j$. Each link transmits traffic in a First-In-First-Out (FIFO) basis. $A_{ij}^k(\omega,\bm{x})$'s are assumed to be non-negative and upper bounded by $A$ for all $\omega$ and $\bm{x}$. 
Based on the definitions above,
the dynamics of the queues can be written following the Lindley recursion:
\begin{align}
Q^{k}_{s_{k}}(t+1)&=[Q^{k}_{s_{k}}(t)+r_k(t)-\sum_{j\in\mathcal{N}_{s_k}}A_{s_kj}^k(\omega(t),\bm{x}(t))]^+, \forall k,\label{eq:dynamics1}\\
Q^k_{d_k}(t)&=0,\quad\forall k,\label{eq:dynamics2}\\
Q_i^k(t+1)&=[Q_i^k(t)+\sum_{j:i\in \mathcal{N}_j}A_{ji}^k(\omega(t),\bm{x}(t))-\sum_{j\in \mathcal{N}_i}A_{ij}^k(\omega(t),\bm{x}(t))]^+,\forall i\neq s_k,d_k. \label{eq:dynamics3}
\end{align}
We define $\Lambda(\omega):=\{(A(\omega,\bm{x}))_{ij}^k, \bm{x}\in \mathcal{X} \}$ as the set of feasible transmission rate vectors under network state $\omega$. Note that the network operator can observe $\Lambda(\omega(t))$ at $t$ but does not know the distribution of $\omega(t)$.\footnote{We assume that $\Lambda(\omega)$ is downward closing in the sense that if $\bm{\lambda} \in\Lambda(\omega)$, then any vector $\bm{\lambda}'$ that equals zero in one coordinate and  equals $\bm{\lambda}$ in all other coordinates is also in $\Lambda$.} Finally, we define $Cap(\mathcal{G})$ as the set of feasible rate vectors $(r_1,\dots,r_K)$, i.e., there exists $\{\lambda(\omega)\}_{ij}^k\in Conv(\Lambda(\omega))$ with
\begin{align*}
&\forall k,\ r_k\le \sum_{\omega\in\mathcal{W}}\sum_{j\in\mathcal{N}_{s_k}}p(\omega)\lambda(\omega)_{ij}^k,\\ &\forall i\in\mathcal{V}, \sum_{\omega\in\mathcal{W}}\sum_{j:i\in \mathcal{N}_j}p(\omega)\lambda(\omega)_{ji}^k \le \sum_{\omega\in\mathcal{W}}\sum_{j\in \mathcal{N}_i}p(\omega)\lambda(\omega)_{ij}^k,
\end{align*}
where $conv(\Lambda(\omega))$ is the convex hull of $\Lambda(\omega)$.
$Cap(\mathcal{G})$ resembles the network capacity region in the traditional infinite-horizon network utility maximization problem \cite{cite:NUM-NeelyFairness}, i.e., the set of traffic rate vectors that can be supported by the network. However, as we consider a finite-horizon setting here, the capacity region here does not exactly characterize the set of job-size vectors that the network can support. Nevertheless, the close connection between the two concepts will be revealed in \textbf{Section \ref{sec:upperbound}}. Furthermore, to prevent trivializing the problem,  we enforce the condition on $Cap(\mathcal{G})$ that it has non-empty interior, i.e., there exists $\eta>0$ such that $(\eta,\ldots,\eta)\in Cap(\mathcal{G})$.

\subsection{Utility Model}
Each job class (user) $k$ is associated with some underlying utility function $f_k$. The utility functions are initially unknown. When a class-$k$ job of size $r_k$ gets delivered to $d_{k}$, we observe and obtain utility of value $f_k(r_k)$. Note that this implies that the utility feedback of each job experiences queueing-style delay, i.e., the time from injecting a job into the network to observing its utility value is equal to the time that the job spends in the network (queues). See Figure \ref{fig:queueingdelay} for further illustration of the feedback delay.
 
For each traffic class $k$, we assume its underlying utility function has the following properties:
\begin{enumerate}
	\item $f_k$ is monotonically non-decreasing and concave.
	\item $f_k$ is bounded on $[0,B]$, i.e., $\forall r\in[0,B], f_k(r)\le D$ for some constant $D$.
	\item $f_k$ is $L$-Lipschitz continuous, i.e., $\forall r_1,r_2\in [0,B]$, $|f_k(r_2)-f_k(r_1)|\le L\cdot|r_2-r_1|$.
\end{enumerate}

\subsection{Problem Formulation}
Given the network $\mathcal{G}$ and time-horizon $T$, we seek to find a scheduling policy that determines the sizes of the jobs sent by the sources and the network actions that maximizes the total utility obtained at the end of the horizon $T$. Formally, let $\Pi$ be the collection of \textit{admissible policies} that make scheduling decisions at time $t$ based on observations obtained before time $t$. Policies in $\Pi$ do not have access to the underlying utility functions or the distribution of network state, but can learn them through observations of utility values and instantiated network state. We further let $\bar{\Pi}$ be the collection of all policies, including non-admissible policies that know the underlying utility functions and the network state distribution. For a policy $\pi$, we define $U(\pi,T)$ to be the total utility obtained from jobs that are delivered by time $T$ under $\pi$. Note that $U(\pi,T)$ is a random variables, the randomness of which comes from the time-varying network state and the (possible) inherent randomness of the scheduling policy. We adopt the notion of regret from the online learning literature as the measure of quality of scheduling policies.
\begin{definition}[Regret]
	The regret of scheduling policy $\pi$ is defined as 
	\[
	R(\pi,T) = \sup_{\pi^*\in \bar{\Pi}}\mathbb{E}[U(\pi^*,T)]-\mathbb{E}[U(\pi,T)],
	\]
\end{definition}
The regret $R(\pi,T)$ measures the gap between the expected utility obtained under $\pi$ and the maximum utility achieved by any (even non-admissible) policy for the given instance.

Based on the above preliminaries, we formally pose the problem of network utility maximization with unknown utility functions, which we will refer to as the L-NUM (Learning-NUM) problem, as one that asks for an admissible scheduling policy with low regret.
\begin{definition}[The L-NUM Problem]
	The L-NUM problem seeks an admissible policy $\pi$ with sublinear regret, i.e., $\lim\limits_{T\rightarrow\infty}\frac{R(\pi,T)}{T}=0$.
\end{definition} 
\textbf{Remark:} (i). A policy that has sublinear regret is \textit{asymptotic optimal}, since the gap between time-average utility achieved by the policy and that of the optimal goes to zero. (ii). Although the regret does not explicitly depend on the queue backlogs at the end of the horizon $T$, the queue backlogs are implicitly accounted for, since the utility $U(\pi,T)$ does not include the jobs that are still in the queue at time $T$.


\begin{figure}
	\centering
	\includegraphics[width=0.8\linewidth]{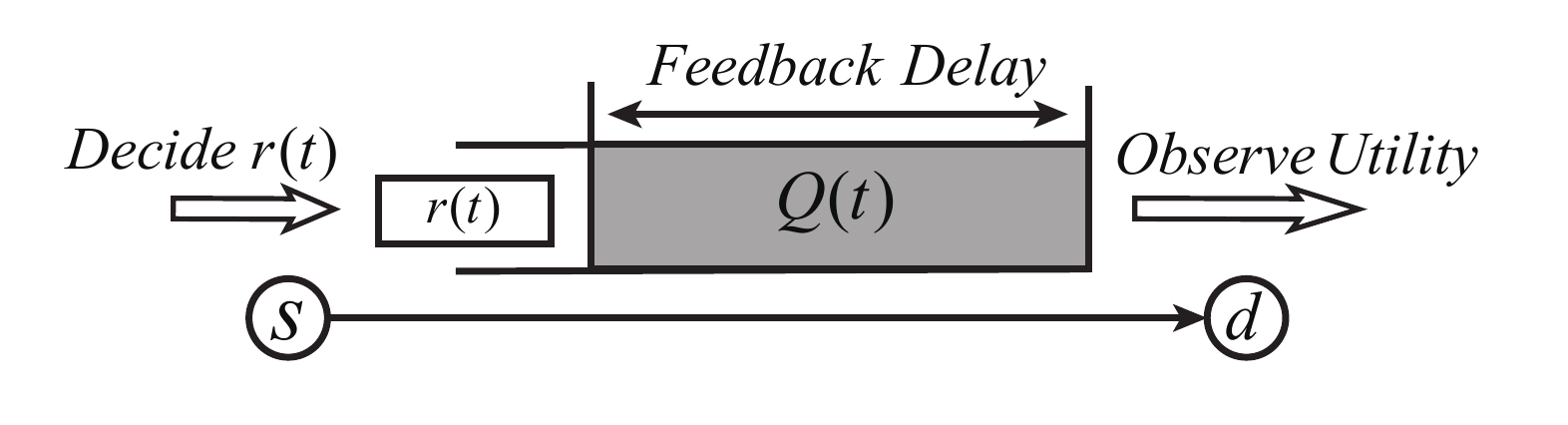}
	\vspace{-4mm}
	\caption{A single-queue example illustrating the queueing-style feedback delay in the L-NUM framework.}
	\vspace{-5mm}
	\label{fig:queueingdelay}
\end{figure}

\section{Upper Bound on the Optimal Utility} \label{sec:upperbound}
If the utility functions are known in advance, L-NUM becomes a finite-horizon stochastic optimization problem.
Typically, the optimal policy for the problem is a dynamic programming-based policy that is intractable and difficult to compare to. Therefore, in this section, we relate the expected utility obtained by the best policy in $\bar{\Pi}$
to the optimal value of a static optimization problem, which motivates the design and analysis of the admissible scheduling policy we propose.  
The optimization problem $\mathcal{P}$ is defined as follows:
\begin{align}
\mathcal{P}:\ &\max_{\{r\}_k}  \sum_{k=1}^Kf_k({r}_k)\\
\textbf{s.t.}\quad & (r_1,\ldots,r_K)\in Cap(\mathcal{G})\label{ieq:capacity}\\
& {r}_k\in [0,B],\quad \forall k. \label{ieq:bound}
\end{align}
Intuitively, the optimization problem characterizes a static version of the L-NUM problem over the job-size variables.
The decision variables $\{r_k\}$'s can be interpreted as average size of jobs of class $k$. $\mathcal{P}$ seeks to maximize the total utility obtained by $\{r_k\}$ such that the vector lies inside the network capacity region. Note that $Cap(\mathcal{G})$ is a convex set over $\{r_k\}$. Hence, $\mathcal{P}$ is a convex optimization problem.

Based on the optimization problem $\mathcal{P}$, we are ready to state the main result of this section, i.e., the optimal value of $\mathcal{P}$ multiplied by the time horizon upper-bounds the maximum expected utility over all policies in $\bar{\Pi}$.
\begin{theorem}\label{thm:upperbound}
	$\sup_{\pi^*\in\bar{\Pi}}\mathbb{E}[U(\pi^*,T)]\le T\cdot OPT(\mathcal{P})$.
\end{theorem}
\begin{proof}
		The main idea of the proof is that, for any given policy, we first take certain averages of the job sizes of each traffic class and then show that the averages satisfy the constraints of $\mathcal{P}$. Next, by the concavity of the underlying utility functions, their corresponding value of the objective function is no less than the expected utility of the policy.
	
	For ease of notations, we prove the theorem for deterministic policies in $\bar{\Pi}$. The case of randomized policies follows similarly.
	Consider an arbitrary policy $\pi^*\in \bar{\Pi}$ and a sample path $\theta$ of its execution on the problem instance. For each traffic class $k$, let $r_k(1,\theta),\ldots,r_k(T,\theta)$ be the size of the jobs specified by $\pi^*$ over the time horizon $T$. Define $\tilde{r}_k(t,\theta)={r}_k(t,\theta)$ if the $t$-th job is delivered by time $T$ and $\tilde{r}_k(t,\theta)=0$ otherwise. Let $\bm{x}(t,\theta)$ be the network action chosen by $\pi^*$ at $t$ and let $\omega(t,\theta)$ be the network state at $t$ under $\theta$. Based on the utility model we have that the utility achieved by $\pi^*$ on sample path $\theta$ is equal to $\sum_{k=1}^K\sum_{t=1}^Tf_k(\tilde{r}_k(t,\theta))$.
	Let $\bar{r}_k(\theta)=\frac{1}{T}\sum_{t=1}^T\tilde{r}_k(t,\theta)$. Since the underlying utility functions are concave, we have
	\begin{align} 
	\sum_{k=1}^K\sum_{t=1}^Tf_k(\tilde{r}_k(t,\theta)) \le T\sum_{k=1}^Kf_k(\bar{r}(\theta)). \label{ieq:jensen1}
	\end{align}
Furthermore, let $\tilde{A}_{ij}^k(\omega(t,\theta),\bm{x}(t,\theta))$ be the realized transmission rate on link $(i,j)$ for class-$k$ at $t$. The realized transmission $\tilde{A}_{ij}^k$ is equal to the offered transmission ${A}_{ij}^k$ when the queue length is greater than the offered transmission, and the realized transmission is smaller otherwise.
From the queue dynamics (Equations \ref{eq:dynamics1}, \ref{eq:dynamics2} and \ref{eq:dynamics3}), we obtain that
	\begin{align}
&\forall k,\quad T\bar{r}_k(\theta)\le\sum_{t=1}^T\sum_{j\in\mathcal{N}_{s_k}}\tilde{A}_{s_kj}^k(\omega(t,\theta),\bm{x}(t,\theta)),\label{ieq:dynamics1}\\
&\forall i\neq s_k,d_k, 	\sum_{t=1}^T\sum_{j:i\in \mathcal{N}_j}\tilde{A}_{ji}^k(\omega(t,\theta),\bm{x}(t,\theta))\le\sum_{t=1}^T\sum_{j\in \mathcal{N}_i}\tilde{A}_{ij}^k(\omega(t,\theta),\bm{x}(t,\theta)).\label{ieq:dynamics2}
	\end{align}
	Define $\hat{p}_{\theta}(\omega),\omega\in\mathcal{W}$ as the empirical distribution of $\omega$,
	\begin{align*}
	\hat{p}_{\theta}(\omega):= \frac{\sum_{t=1}^T\mathbbm{1}\{\omega(t,\theta)=\omega \}  }{T}.
	\end{align*}
	It follows from (\ref{ieq:dynamics1}), (\ref{ieq:dynamics2}) that for each $\omega\in\mathcal{W}$, there exists $(\tilde{\lambda}(\omega,\theta))_{ij}^k \in Conv(\Lambda(\omega))$ such that
		\begin{align*}
&\forall k,\quad	\bar{r}_k(\theta)\le\sum_{\omega\in\mathcal{W}} \sum_{j\in\mathcal{N}_{s_n}} \hat{p}_{\theta}(\omega)\tilde{\lambda}_{s_kj}^k(\omega,\theta),\\
&\forall i\neq s_k,d_k,	\sum_{\omega\in\mathcal{W}} \sum_{j:i\in\mathcal{N}_j} \hat{p}_{\theta}(\omega)\tilde{\lambda}_{ji}^k(\omega,\theta) \le \sum_{\omega\in\mathcal{W}} \sum_{j\in\mathcal{N}_i} \hat{p}_{\theta}(\omega)\tilde{\lambda}_{ij}^k(\omega,\theta).
	\end{align*}
	Moreover, as $\Lambda(\omega)$ is downward-closing, we further have that there exists $({\lambda}(\omega,\theta))_{ij}^k \in Conv(\Lambda(\omega))$ such that
	\begin{align*}
	&\forall k,\quad	\bar{r}_k(\theta)=\sum_{\omega\in\mathcal{W}} \sum_{j\in\mathcal{N}_{s_k}} \hat{p}_{\theta}(\omega){\lambda}_{s_kj}^k(\omega,\theta),\\
	&\forall i\neq s_k,d_k,\quad	\sum_{\omega\in\mathcal{W}} \sum_{j:i\in\mathcal{N}_j} \hat{p}_{\theta}(\omega){\lambda}_{ji}^k(\omega,\theta) = \sum_{\omega\in\mathcal{W}} \sum_{j\in\mathcal{N}_i} \hat{p}_{\theta}(\omega){\lambda}_{ij}^k(\omega,\theta).
	\end{align*}

	Taking expectation over $\theta$, we have $(\mathbb{E}_\theta[\bar{r}_1(\theta)],\ldots,\mathbb{E}_\theta[\bar{r}_1(\theta)])\in Cap(\mathcal{G})$.
	Moreover, it is easy to see that $0\le \mathbb{E}_\theta[\bar{r}_k(\theta)] \le B$ for all $k$. Therefore, the vector $(\mathbb{E}_\theta[\bar{r}_1(\theta)],\ldots,\mathbb{E}_\theta[\bar{r}_1(\theta)])$ is feasible to $\mathcal{P}$.
	 Hence, $OPT(\mathcal{P})\ge \sum_{k=1}^Kf_k(\mathbb{E}_\theta[\bar{r}_k(\theta)])$.
	Invoking the concavity of $f_k$'s again, by Jensen's inequality, we have
	for all $k$, $f_k(\mathbb{E}_\theta[\bar{r}_k(\theta)]) \ge \mathbb{E}_{\theta}[f_k(\bar{r}_k(\theta))]$. Combining this with (\ref{ieq:jensen1}), we obtain
	\begin{align}
&	OPT(\mathcal{P})\ge \sum_{k=1}^Kf_k(\mathbb{E}_\theta[\bar{r}_k(\theta)]) \ge \mathbb{E}\left[\sum_{k=1}^K f_k(\bar{r}_k(\theta))\right] \nonumber\\\ge& \frac{1}{T}\mathbb{E}\left[ \sum_{k=1}^K\sum_{t=1}^Tf_k(\tilde{r}_k(t,\theta)) \right],
	\end{align}
	which concludes the proof.
\end{proof}	
It is worth pointing out that Theorem \ref{thm:upperbound} does not imply that the optimal policy is a static one that assigns the job sizes according to the solution to the optimization problem $\mathcal{P}$. Such a policy would not achieve an expected utility of $T\cdot OPT$ since the expected number of jobs delivered is typically less than $T$. Indeed, due to the stochastic network dynamics, a portion of the jobs will still remain in the queues by the end of the time horizon.
Despite that, the theorem does provide the insight that a policy achieves low regret if it can closely approximate the solution to $\mathcal{P}$ at each time slot. As the objective function of $\mathcal{P}$ is unknown, the problem has similar flavor to online/zeroth-order optimization \cite{cite:oco,cite:banditiot,cite:banditcomplexity}. However, in the L-NUM problem we are facing two additional challenges. First, the feasibility region in the L-NUM is stochastic and not explicitly given as the distribution of network states is unknown. Thus, we cannot rely on method that requires the feasibility region to be known in advance \cite{cite:convexbandit}.
Second, the queueing-style delay of the feedback compromises the policy's ability to adjust based on utility observations. As the delay is action-dependent and may be unbounded, it also poses more stringent requirement on controlling the network queue lengths.

 \section{Online Scheduling Policy}\label{sec:policy}
 In this section, we introduce the scheduling policy we propose for the L-NUM framework -- the Parallel Gradient Sampling Max-Weight (P-GSMW) policy. The P-GSMW policy is composed of embedding an algorithm (called Gradient Sampling Max-Weight, GSMW) that makes job-size and scheduling decisions based on immediate feedback (no delay) into a parallel-instance paradigm that handles the feedback delay.
The GSMW algorithm essentially combines the ideas of drift-plus-penalty optimization \cite{cite:neely1}, gradient sampling \cite{cite:flaxman}, and Max-Weight scheduling. The parallel-instance paradigm invokes multiple parallel instances of the GSMW algorithm such that each instance essentially runs in a no-delay setting.
In the following, we first introduce the GSMW algorithm, and then combine it with the parallel-instance paradigm. Finally, we provide discussion on the challenges posed by the feedback delay.
   
   \subsection{The GSMW Algorithm}   
   In the presentation of the GSMW algorithm, we assume a no-delay setting, i.e., the utility values of the jobs can be observed immediately after job-size decision. We will handle the feedback delay with the parallel-instance paradigm in subsequent sections.
   
  The GSMW algorithm (\textbf{Algoritm \ref{alg:policy}}) maintains a virtual job size variable $\hat{r}_k$ for each class $k$ and utilizes queue lengths to update the virtual job size variables and network actions. The $\hat{r}_k$'s are updated once every two slots,
  which essentially divides the time horizon into epochs of size two (without loss of generality, we assume the horizon $T$ to be even). 
  For simplicity of notations, we will assume that the network state remains unchanged for each epoch and refer to an epoch as a time slot indexed by $t \in \{1,\ldots, T\}$ for the rest of the paper, i.e., at each slot, we need to make scheduling decision and job-size decision for two incoming jobs of each class.\footnote{This assumption is purely made for notational convenience. Our results can be straightforwardly adapted to the original setting without the assumption.}

  At each slot $t\in \{1,\ldots,T\}$, the network action is chosen according to a Max-Weight-like rule (Line \ref{alg:maxweight}).
   The decisions on job size are made based on the virtual job size variables at the corresponding epoch.  
  The updates of virtual job size variables are determined by gradient estimates of the utility functions and queue lengths. Since the utility functions are unknown, GSMW constructs the gradient estimates using observations of function values. Specifically, at slot $t$, each source $s_k$ injects a first job of size $\hat{r}_k(\tau)+\delta$ and a second job of size  $\hat{r}_k(\tau)-\delta$  for each $k\in n$ and obtains the feedback of value $f_k(\hat{r}_k(\tau)+\delta)$ and $f_k(\hat{r}_k(\tau)-\delta)$ (Lines \ref{alg:sample1}, \ref{alg:sample2}).  The two feedback values obtained are combined to form the gradient estimate of $f_k$ at $\hat{r}_k(\tau)$ (Line \ref{alg:gradient}). The gradient estimate is then fed into the update of the virtual variable $\hat{r}_k$ (Line \ref{alg:update}). The projection step  $\mathcal{P}_{[\delta,B-\delta]}$ of Line \ref{alg:update}, defined as the projection on to interval $[\delta,B-\delta]$ by the Euclidean norm, is to ensure that $\hat{r}_k(\tau)+\delta$ and $\hat{r}_k(\tau)-\delta$ always lie in the domain $[0,B]$.  Here, parameter $V$ controls the relative weights of gradient and queue length while parameter $\alpha$ determines the step size.


   \begin{algorithm}
   	\caption{The Gradient Sampling Max-Weight Algorithm}
   	\begin{algorithmic}[1] \label{alg:policy}		
   		\REQUIRE{Network $\mathcal{G(V,E)}$, parameters $V,\delta,\alpha$}
   		
   		\STATE \textbf{Initialize: } $\bm{x}(0)\in\mathcal{X}, \hat{r}_k(0)=\delta$.\\
   		\FOR{$t=1,2,\ldots,T$}
   			\STATE
   		${\bm{x}}(t):=\arg\max_{\bm{x}\in\mathcal{X}}\sum_{i\in V}\sum_{k=1}^K{A}_{ij}^k(\omega(t),\bm{x})[Q_i^k(t)-Q_j^k(t)] $ \label{alg:maxweight}
   		\FOR{$k=1,\ldots,K$}
   		\STATE $s_{k}$ injects job of size $\hat{r}_k(t)+\delta$ and observes ${f}_k(\hat{r}_k(t)+\delta)$. \label{alg:sample1}\\
   		\STATE $s_{k}$ injects job of size $\hat{r}_k(t)-\delta$ and observes ${f}_k(\hat{r}_k(t)-\delta)$. \label{alg:sample2}\\
   		\STATE $\hat{\nabla}f_k(\hat{r}_k(t)):=\frac{{f}_k(\hat{r}_k(t)+\delta)-{f}_k(\hat{r}_k(t)-\delta)}{2\delta}$\label{alg:gradient}\\
   		\ENDFOR
   		\STATE
   		Update queue lengths according to $r_k(t),\bm{x}(t)$.
   		\FOR{$k=1,\ldots,K$}
   			\STATE
   		$\hat{r}_k(t+1)$\\$:=\mathcal{P}_{[\delta,B-\delta]}\left[ \hat{r}_k(t)+\frac{1}{\alpha} (V\cdot \hat{\nabla}f_k(\hat{r}_k(t))-Q^{k}_{s_{k}}(t))  \right]$\label{alg:update}\\
   		\ENDFOR
 
   		\ENDFOR
   	\end{algorithmic}
   \end{algorithm}

 \subsection{The Parallel-Instance Paradigm}
  
  In order to handle the feedback delay, we design a parallel-instance paradigm that encapsulates the GSMW algorithm, and forms the Parallel-instance GSMW (P-GSMW) policy. The details of the P-GSMW policy are shown in \textbf{Algorithm \ref{alg:p-gsmw}}. Similar to the GSMW algorithm, the network action $\bm{x}(t)$ at each time slot is still determined by the Max-Weight rule. The key difference is that the paradigm maintains a set  of parallel instances of the GSMW algorithm, which we will refer to as the instance reservoir $\mathcal{I}$. Each instance can be in one of the two possible status: FRESH and STALE. FRESH status means that the instance has obtained the corresponding utility feedback and can perform updates on the virtual job-size variables (Line \ref{alg:update_p} of \textbf{Algorithm \ref{alg:p-gsmw}}); STALE status means that the instance is still waiting for utility feedback. We use $\{r^I_k(t)\}$ to denote the virtual job size variables maintained by instance $I$.
  When we need to make job size decisions, if there is a FRESH instance available in the reservoir, we ``invoke'' the instance by performing updates and deciding on the job sizes based on the updated virtual job-size variables of the instance (Line \ref{alg:update_p} of \textbf{Algorithm \ref{alg:p-gsmw}}). If there are multiple FRESH instances, we select one arbitrarily. We then change the instance's status to STALE (Line \ref{alg:status}). If there is no FRESH instance available, we initialize a new instance, add it to $\mathcal{I}$ (Lines 9 and 10). The virtual job-size variables of instances that are not invoked remains unchanged (Line \ref{alg:same}).
   Upon delivery of jobs, we observe utility values and feed them to the corresponding instances. If an instance has all the utility observations available for the jobs injected when it was last invoked, we change its status from ``STALE'' to ''FRESH'' (Line \ref{alg:collect}).  We further illustrate the parallel-instance paradigm with a single-queue example in Figure \ref{fig:pgsmw}.
   \begin{figure}
   	\centering
   	\includegraphics[width=0.8\linewidth]{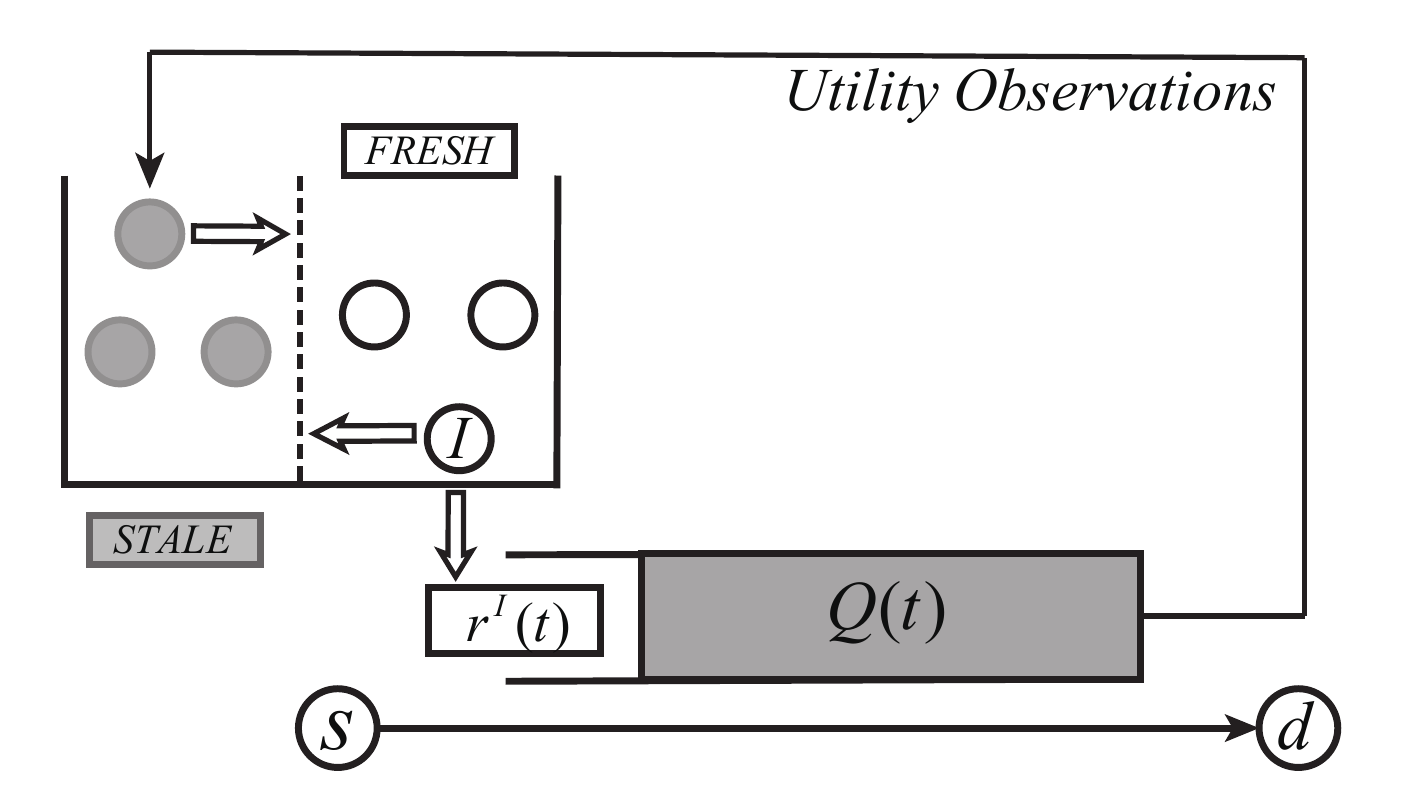}
   	\vspace{-4mm}
   	\caption{A single-queue example illustrating the parallel-instance paradigm.}
   	\vspace{-3mm}
   	\label{fig:pgsmw}
   \end{figure}
   
 \textbf{Remark:} Our parallel-instance paradigm has a similar flavor to the technique in \cite{cite:delayfeedback1} for online learning with delayed feedback. However, the observation delay in the L-NUM framework may be unbounded and is action-dependent, which is more general than the bounded, decision-independent delay considered in \cite{cite:delayfeedback1}.

    \begin{algorithm}
 	\caption{The Parallel-instance GSMW Policy}
 	\begin{algorithmic}[1] \label{alg:p-gsmw}		
 		\REQUIRE{Network $\mathcal{G(V,E)}$, parameters $V,\delta,\alpha$, instance reservoir $\mathcal{I}$}
 		
 		\FOR{$t=1,2,\ldots,T$}
 		\STATE
 		${\bm{x}}(t):=\arg\max_{\bm{x}\in\mathcal{X}}\sum_{i\in V}\sum_{k=1}^K{A}_{ij}^k(\omega(t),\bm{x})[Q_i^k(t)-Q_j^k(t)] $ \label{alg:maxweight_p}
 		\IF{There exists a FRESH instance $I_t\in\mathcal{I}$}
 		\FOR{$k=1,\ldots,K$}
 		\STATE $\hat{r}^{I_t}_k(t)$\\$:=\mathcal{P}_{[\delta,B-\delta]}\left[ \hat{r}^{I_t}_k(t-1)+\frac{1}{\alpha} (V\cdot \hat{\nabla}f_k(\hat{r}^{I_t}_k(t-1))-Q^{k}_{s_{k}}(t))  \right]$\label{alg:update_p}
 		\STATE $s_{k}$ injects job of size $\hat{r}^{I_t}_k(t)+\delta$ and another job of size $\hat{r}^{I_t}_k(t)-\delta$. \label{alg:sample_p}\\
 		\STATE Change the status of $I_t$ to STALE. \label{alg:status}
 		\ENDFOR
 		\ELSE 
 		\STATE Create a new instance $I_t$ \COMMENT{No FRESH instance in $\mathcal{I}$.}
 		\STATE For each $k$, initialize $\hat{r}_k^{I_t}(t):=\delta$, and $s_{k}$ injects job of size $\hat{r}^{I_t}_k(t)+\delta$ and another job of size $\hat{r}^{I_t}_k(t)-\delta$
 		\ENDIF
 		\STATE
 		Update queue lengths according to $r_k(t),\bm{x}(t)$.
 		\STATE
 		$\{\hat{r}^J_k(t) \}:=\{\hat{r}^J_k(t) \}$ for $J\in \mathcal{I}, J\neq I_t$. \label{alg:same}
 		\STATE Collect utility observations from delivered jobs and form gradient estimates
 		$\hat{\nabla}f_k(\hat{r}^I_k(t)):=\frac{{f}_k(\hat{r}^I_k(t)+\delta)-{f}_k(\hat{r}^I_k(t)-\delta)}{2\delta}$
 		\FOR{STALE instance $I\in \mathcal{I}$}
 		\STATE Change the status of $I$ to FRESH if it has obtained all outstanding gradient estimates. \label{alg:collect}
 		\ENDFOR
 		\ENDFOR
 	\end{algorithmic}
 \end{algorithm}


 \subsection{Policy Analysis}
 In this section, we analyze the regret achieved by the P-GSMW policy $\pi_{P-GSMW}$. The main result is presented Theorem {\ref{thm:regret}}.
 
 \begin{theorem}\label{thm:regret}
 	 The Parallel-instance Gradient Sampling Max-Weight policy $\pi_{P-GSMW}$ achieves $\tilde{O}(T^{3/4})$ regret by setting $\alpha = 2K\sqrt{T}/\eta, V=T^{1/4}, \delta = T^{-1/2}$., i.e., 
 	\[
 	R(\pi_{P-GSMW},T) = \sup_{\pi^*\in \bar{\Pi}}\mathbb{E}[U(\pi^*,T)]-\mathbb{E}[U(\pi_{P-GSMW},T)] = \tilde{O}(T^{3/4}).
 	\]	
 \end{theorem}
\textbf{Remark:} It can be shown that without feedback delay, by setting $\alpha = O(T), V = O(\sqrt{T}), \delta = O(1/\sqrt{T})$ (rather than $\alpha = O(\sqrt{T}), V = O(T^{1/4}), \delta = O(1/\sqrt{T})$ in Theorem \ref{thm:regret}), the GSMW algorithm achieves a regret of order $\tilde{O}(\sqrt{T})$, which matches the established regret lower bound $\Omega(\sqrt{T})$ \cite{cite:convexbandit}. Under the queueing-style feedback delay, the P-GSMW policy achieves $\tilde{O}({T}^{3/4})$ which is higher than $\tilde{O}(\sqrt{T})$. This raises the question whether the delay of L-NUM fundamental increases the difficulty of the problem, i.e., a lower bound better than $\Omega(\sqrt{T})$ can be shown, or that there exists algorithm for L-NUM that has regret better than $\tilde{O}({T}^{3/4})$. We leave this as a future direction.
\begin{proof}(of Theorem \ref{thm:regret})
  As we focus on bounding the regret with respect to the time horizon $T$, we will use $C$ to represent a generic constant that does not depend on $T$. Note that $C$ may depend on parameters such as $A,B,D,L$, and the $C$'s that appear in different equations might not be equal. In this section, instead of directly analyzing the P-GSMW policy, we will analyze the GSMW algorithm (\textbf{Algorithm \ref{alg:policy}}) in a no-delay setting, and illustrate how to extend the analysis to the P-GSMW policy in the end. A complete analysis of the P-GSMW policy can be found in Appendix \ref{app:policyproof}.
  
  Recall that in the no-delay setting, we can observe the utility value immediately after the job-size decision. Thus, we can apply the GSMW algorithm $\pi_{GSMW}$ with the same parameter values as indicated in Theorem \ref{thm:regret}. We will show that in this case, the GSMW algorithm achieves $\tilde{O}(T^{3/4})$ regret.
  We first decompose the regret into two components: one incurred through incrementally solving $\mathcal{P}$ (\textit{Utility Regret}) and the other caused by the undelivered jobs in the queues at the end of the time horizon (\textit{Queueing Regret}). 
  \begin{lemma}\label{lemma:decompose}
  	Let $\{r^*\}_k$ be the optimal solution to $\mathcal{P}$,
  	\begin{align*}
  	&R(\pi_{GSMW},T) \\\le & 2\mathbb{E}\left[\sum_{t=1}^{T}\sum_{k=1}^Kf_k(r^*_k)-f_k(\hat{r}_k(t)) \right]+C\sum_{i\in V}\sum_{k=1}^K\mathbb{E}[Q_i^k(T)]+CT\delta. 
  	\end{align*}
  \end{lemma}
 \textit{Proof Sketch:}
 	By definition, the utility achieved by $\pi_{GSMW}$ is equal to $\sum_{t=1}^T\sum_{k=1}^Kf_k(\hat{r}_k(t)+\delta)+f_k(\hat{r}_k(t)-\delta)$ minus the total utility of (undelivered) jobs in the queue at time $T$. By Theorem \ref{thm:upperbound} and that the utility of a single job is upper bounded by D, i.e., the upper bound of the utility function, we can bound the regret by the RHS of Lemma 1, where the first term (utility regret) accounts for the cumulative difference between the $\sum_{k}^Kf_k(\hat{r}_k(t))$ and $OPT(\mathcal{P})$, the second term (queueing regret) accounts for the jobs in the queue, and the third term comes from that the sizes of jobs injected by the sources are $\delta$-away from the virtual job size variables.

By taking $\delta = T^{-1/2}$, we have $CT\delta=O(\sqrt{T})$. To prove Theorem \ref{thm:regret}, we can thus proceed to bound the queueing regret $\sum_{i\in V,k}\mathbb{E}[Q_i^k(T)]$ and the utility regret $\mathbb{E}\left[\sum_{t=1}^{T}\sum_{k=1}^Kf_k(r^*_k)-f_k(\hat{r}_k(t))\right]$. In the following, we will first show that the total queue length of the network (and thus the queue length regret) under the GSMW algorithm is of order $\tilde{O}(\sqrt{T})$. Using this, we will then show that the utility regret is of order $\tilde{O}(T^{3/4})$.

The GSMW algorithm controls the utility regret and queueing regret through the updates of job-size variables (Line \ref{alg:update} of \textbf{Algorithm \ref{alg:policy}}). The term  $V\cdot \hat{\nabla}f_k(\hat{r}_k(t))$ moves the variables towards optimizing utilities, and the $-Q^{k}_{s_{k}}(t)$ part aims at controlling queue lengths while $\alpha$ controls the step size. For any $\{r \}_k$ with $r_k\in[\delta, B-\delta]$, by expanding the term $\sum_{k=1}^K(\hat{r}_k(t+1)-r_k)^2$ using Line \ref{alg:update}, we have the following lemma.
\begin{lemma}\label{lemma:update1}
		For each $t$, for any $\{r \}_k$ with $r_k\in[\delta, B-\delta]$
		\begin{align*}
		&\sum_{k=1}^K \left[V\hat{\nabla}f_k(\hat{r}_k(t))(r_k-\hat{r}_k(t))\right] +\sum_{k=1}^K\left[Q^{k}_{s_{k}}(t)\hat{r}_k(t)  \right]\\
		\le & \sum_{k=1}^K\left[Q^{k}_{s_{k}}(t)r_k+\alpha[(\hat{r}_k(t)-r_k)^2-(\hat{r}_k(t+1)-r_k)^2]+{C}.   \right]
		\end{align*}
\end{lemma}

From the Max-Weight selection rule of network action (Line \ref{alg:maxweight_p} of \textbf{Algorithm \ref{alg:p-gsmw}}), we have the following lemma.
\begin{lemma}\label{lemma:update2}
		At every time slot $t$, for any $\bm{x}\in \mathcal{X}$, and for all $\{r\}_k$
\begin{align*}
&\sum_{k=1}^KQ^{k}_{s_{k}}(t)\left[{r}_k-\sum_{j\in\mathcal{N}_{s_k}}A_{s_kj}^k(\omega(t),{\bm{x}}(t))\right]\\&\quad+\sum_{k=1}^K\sum_{i\in V,i\neq s_k}Q_i^k(t)\left[\sum_{j:i\in \mathcal{N}_j}{A}_{ji}^k(\omega(t),{\bm{x}}(t))-\sum_{j\in \mathcal{N}_i}{A}_{ij}^k(\omega(t),{\bm{x}}(t))\right]
\\\le & \sum_{k=1}^KQ^{k}_{s_{k}}(t)\left[{r}_k-\sum_{j\in\mathcal{N}_{s_k}}A_{s_kj}^k(\omega(t),{\bm{x}})\right]\\&\quad+\sum_{k=1}^K\sum_{i\in V,i\neq s_k}Q_i^k(t)\left[\sum_{j:i\in \mathcal{N}_j}{A}_{ji}^k(\omega(t),{\bm{x}})-\sum_{j\in \mathcal{N}_i}{A}_{ij}^k(\omega(t),{\bm{x}})\right]
\end{align*}
\end{lemma}

Lemmas \ref{lemma:update1} and \ref{lemma:update2} lay the foundation of the analysis of queueing regret and utility regret.

\subsubsection{Queueing Regret}
We write the vector of queue lengths at time $t$ as $\bm{Q}(t)$.
The key to bounding the queueing regret is to bound the quadratic drift of $\bm{Q}(t)$. The drift consists of terms involving the queues at the source $\{Q_{s_k}^k\}$, and the queues in the network $\{Q_i^k\}$. We use Lemma \ref{lemma:update1} to handle the former, and use Lemma \ref{lemma:update2} for the latter.
\begin{lemma}\label{lemma:drift1}
	There exists $\epsilon>0$ such that under the GSMW algorithm, for all $t\le T$,
\begin{align*}
\mathbb{E}[||\bm{Q}(t+1)||^2-||\bm{Q}(t)||^2\mid \bm{Q}(t)] \le -\epsilon  \sum_{i\in V}\sum_{k=1}^KQ_i^k(t) + C\sqrt{T}.
\end{align*}
\end{lemma}
\textit{Proof Sketch:} The drift argument follows from Lemmas \ref{lemma:update1} and \ref{lemma:update2} by taking $\{r\}_k$ therein to be the vector $(\delta,\ldots,\delta)$ and utilizing the Slater's condition. 

Based on Lemma \ref{lemma:drift1}, we use a result on stochastic processes with negative drift from \cite{cite:neely1}, which leads to Proposition \ref{prop:workload1} that essentially concludes the analysis of the queueing regret.
\begin{proposition}\label{prop:workload1}
	Under GSMW, $\forall t\le T$, $\mathbb{E}[\sum_{i\in V}\sum_{k=1}^KQ_i^k(t)]\le \tilde{O}(\sqrt{T})$.
\end{proposition}

\subsubsection{Utility Regret}
 It can be shown that the gradient estimate of the GSMW algorithm $\hat{\nabla}f_k$ is equal to the gradient of a smoothed version of $f_k$ defined as $\tilde{f}_k(r)=\frac{1}{2\delta}\int_{-\delta}^{\delta}f_k(r+z)\mathrm{d}z$. Note that by definition $\tilde{f}_k$ is also concave and Lipschitz-continuous. Moreover, by the concavity and Lipschitz continuity of ${f}_k$, for all $r\in[0,B]$, $f_k(r)-C\delta\le \tilde{f}_k(r) \le f_k(r)$. Hence, by using $\hat{\nabla}f_k$ as gradients, we are essentially optimizing with respect to objective function $\sum_{k=1}^K\tilde{f}_k(r) $, which is at most $C\delta$ away from the true objective function $\sum_{k=1}^K {f}_k(r) $. Accumulating over $T$ time slots, such approximation contributes to at most $O(\sqrt{T})$-regret.
 Next, from Lemma \ref{lemma:update1}, we take $\{r\}_k$ to be the optimal solution $\{r^*\}_k$ to the optimization problem $\mathcal{P}$ and use $V\hat{\nabla}f_k(\hat{r}_k(t))(r^*_k-\hat{r}_k(t))\ge V(\tilde{f}_k(r^*_k)-\tilde{f}_k(\hat{r}_k(t)))$. This will us a handle on the utility regret.
Finally, we bound the utility regret by summing over the inequality of Lemma \ref{lemma:update1}. Note that the term $\alpha[(\hat{r}_k(t)-r_k)^2-(\hat{r}_k(t+1)-r_k)^2]$ telescopes, and the terms involving $Q_{s_k}^k (t)$'s can be bounded by previous results on the queueing regret (Proposition \ref{prop:workload1}). Plugging in the parameter values, we can show that the utility regret is of order $\tilde{O}(T^{3/4})$. Combining the analysis of queueing regret, utility regret and Lemma \ref{lemma:decompose} concludes the proof of Theorem \ref{thm:regret}.

\subsubsection{Extension to P-GSMW} The analysis of the P-GSMW policy essentially follows the same vein. Lemmas \ref{lemma:decompose},\ref{lemma:update1} and \ref{lemma:update2} still hold by replacing $\hat{r}_k(t)$ by $\hat{r}_k^{I_t}(t)$, which is the job-size variables used by the P-GSMW policy at time $t$. Note that here $I_t$ denotes the instance used at $t$, and $I_t$ may be different at different $t$. Using \ref{lemma:update1} and \ref{lemma:update2}, the analysis of queueing regret can be carried out in a similar way for the P-GSMW policy, as it does not need  $\hat{r}_k^{I_t}(t)$'s to come from the same instance at each time but only relies on them being bounded. Therefore, we can establish that under the P-GSMW policy, the queueing regret is still of order $\tilde{O}(\sqrt{T})$. 

Proceeding to the utility regret, most of the previous reasoning still holds for the P-GSMW policy, except that the term that corresponds to $\alpha[(\hat{r}_k(t)-r_k)^2-(\hat{r}_k(t+1)-r_k)^2]$ becomes $\alpha[(\hat{r}^{I_t}_k(t-1)-r_k)^2-(\hat{r}^{I_t}_k(t)-r_k)^2]$, which no longer telescopes since $I_t$ may change with $t$. Instead, it only partially telescopes, leading to a term of $O(\alpha |\mathcal{I}|)$ where $|\mathcal{I}|$ is the total number of instances in the reservoir by the end of the time horizon. This also reflects the impact of having multiple instances in the P-GSMW: each instance get updated for fewer than $T$ times, which makes the job-size variables $\{r^I_k(t)\}$ of each instance converge slower to the optimal compared to the GSMW algorithm in the no-delay setting.
To bound the term $\alpha |\mathcal{I}|$, we use the previously obtained result on queueing regret, relate $|\mathcal{I}|$ to the maximum total queue lengths of the network, and show that the utility regret is of order $\tilde{O}(T^{3/4})$.
\end{proof}

\subsubsection{Challenges Posed By Feedback Delay}
As we have shown in the preceding analysis, the update of job-size variables in the GSMW algorithm (Line \ref{alg:update} of \textbf{Algorithm \ref{alg:policy}}) is composed of one term $V\cdot \hat{\nabla}f_k(\hat{r}_k(t))$ that approximates the gradient and moves the variables towards optimizing utilities, and another term $-Q^{k}_{s_{k}}(t)$ that aims at controlling queue lengths. 
With the presence of feedback delay, GSMW is not applicable as we can only observe the utility values used in the gradient approximation (Line \ref{alg:gradient}) after the jobs get delivered to the destination. Therefore, we may not have the first term available when we perform the updates of GSMW. While we have shown that such difficulty can be overcome by the parallel-instance paradigm, one natural question is \textbf{would other simpler adaptation of the GSMW algorithm work?}. We will now look at two other more straightforward modifications of the GSMW algorithm. By arguing at an intuitive level that they are unlikely to work, we further justify the necessity of having the parallel-instance paradigm.

One possible alternative is to use an ``episodic approach'', i.e., keep the job-size variables unchanged (for an episode of multiple slots) until the utility observations needed become available, and update job-size variables once every episode. Here, as the delay of jobs is essentially proportional to the queue lengths, the length of episodes need to be set to be large enough as the maximum queue length throughout the optimization process. However, this would cause the sizes of the jobs (traffic injected to the network) not be able to adjust timely with respect to the queue lengths (as Line \ref{alg:update} only gets executed once every episode), which will further lead to larger queue lengths (episode length), and thus create a ``feedback loop'' that makes the algorithm suffer from linear regret.

Another possible method is to use old gradients, e.g., we execute Line \ref{alg:update} every time slot, but use the most recently available gradient estimates. This makes the algorithm adjusts according to queue length every time slot, and thus can maintain the queue length bound of GSMW. However, the feedback delay (queue lengths) is still non-trivial and cannot be bounded by a constant independent of $T$, which will result in a large bias in the gradient estimates used in the updates and lead to linear regret.

\section{Applications} \label{sec:application}
In this section, we apply our L-NUM framework to example applications including database query, job scheduling and video streaming. 
\subsection{Database Query} \label{sec:database}
In this example, we consider a setting where there are $K$ users $\{u_1,\ldots,u_K\}$ querying a central database.\footnote{Note that in this example, we only consider the access to the database and not the problem of routing the queries trough the network.} At each time $t$, user $u_k$ issues a query of size $r_k$ to the central database, with $r_k$ representing the processing requirement of the query. The issued queries get buffered in the queue of the database and the database can process $c$ unit of requests in a first-come-first-serve order at each time slot.
Each use $u_k$ is associated with an underlying utility function $f_k$ that captures the relationship between the processing requirement and utility gained from the query. $f_k$ is Lipschitz continuous and concave, which reflects the diminishing return property of query processing. Over a time horizon of $T$, the goal is to maximize the total utility of the processed queries.

Applying our framework to the database query example, the network is a simple one with a single state, one source node, one destination node and a link between them (See Figure \ref{fig:database}). All the users are mapped to the source node and the database corresponds to the link with the transmission rate of the link at each time slot being equal to the processing capacity $c$ of the database. The network action component of the framework is not needed. The queue at the source node, corresponding to the buffer of the database, buffers the jobs (query requests) of all users. P-GSMW policy adjusts the size of the query according to the gradient estimates and the queue size at the source node, and achieves $\tilde{O}(T^{3/4})$-regret.

\begin{figure}
	\centering
	\includegraphics[width=0.8\linewidth]{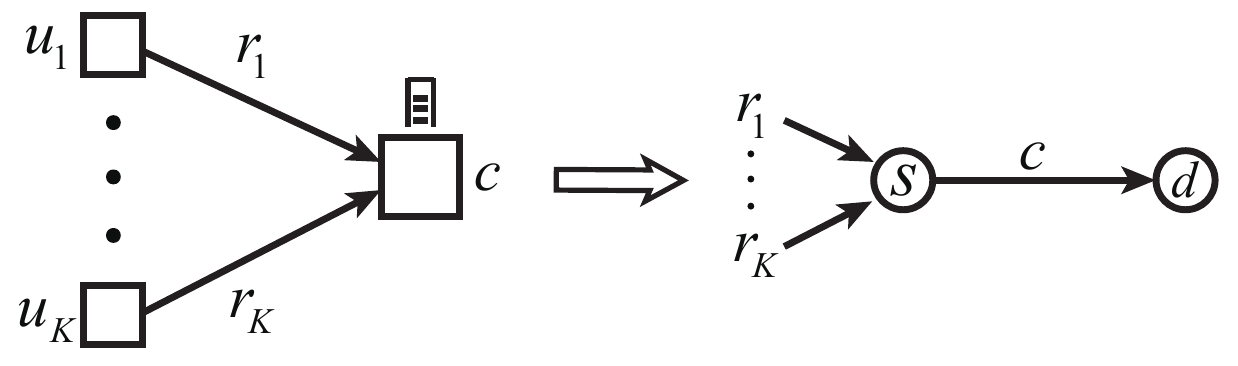}
	\vspace{-3mm}
	\caption{Correspondence between database query and the L-NUM framework.}
	\vspace{-3mm}
	\label{fig:database}
\end{figure}

\subsection{Job Scheduling} \label{sec:jobschedule}
Consider a discrete-time system with with a set of job schedulers (dispatchers) $\{u_1,\ldots, u_K\}$ and a set of parallel servers $\{s_1,\ldots,s_M\}$ that form a bipartite graph. We use $S_{u_k}$ to denote the set of servers that dispatcher $u_k$ is connected to. At each time, a class-$k$ job arrives at the dispatcher $u_k$ and the dispatcher sends the job to one of the servers in $S_{u_k}$ for execution. The job dispatcher also determines the resource (e.g. computation, memory) requirement of each job. Each server $s_m$ can provide $c_m(t)$ amount of resources at time $t$ with $c_m(t)$ being a sequence of i.i.d. discrete random variables. Class-$k$ jobs have underlying utility function $f_k$. A utility of $f_k(r_k)$ is obtained when a class-$k$ job of resource requirement $r_k$ is completed at a server. We seek a scheduling policy that determines the resource requirement and target server of each job. The goal is to maximize the total utility gained from jobs completed over the time horizon $T$. This example particularly mirrors applications where the jobs are flexible in terms of resource requirement (e.g., model training for machine learning tasks in cloud computing \cite{cite:MLtask,cite:NUM-Job1}).

\begin{figure}
	\centering
	\includegraphics[width=0.8\linewidth]{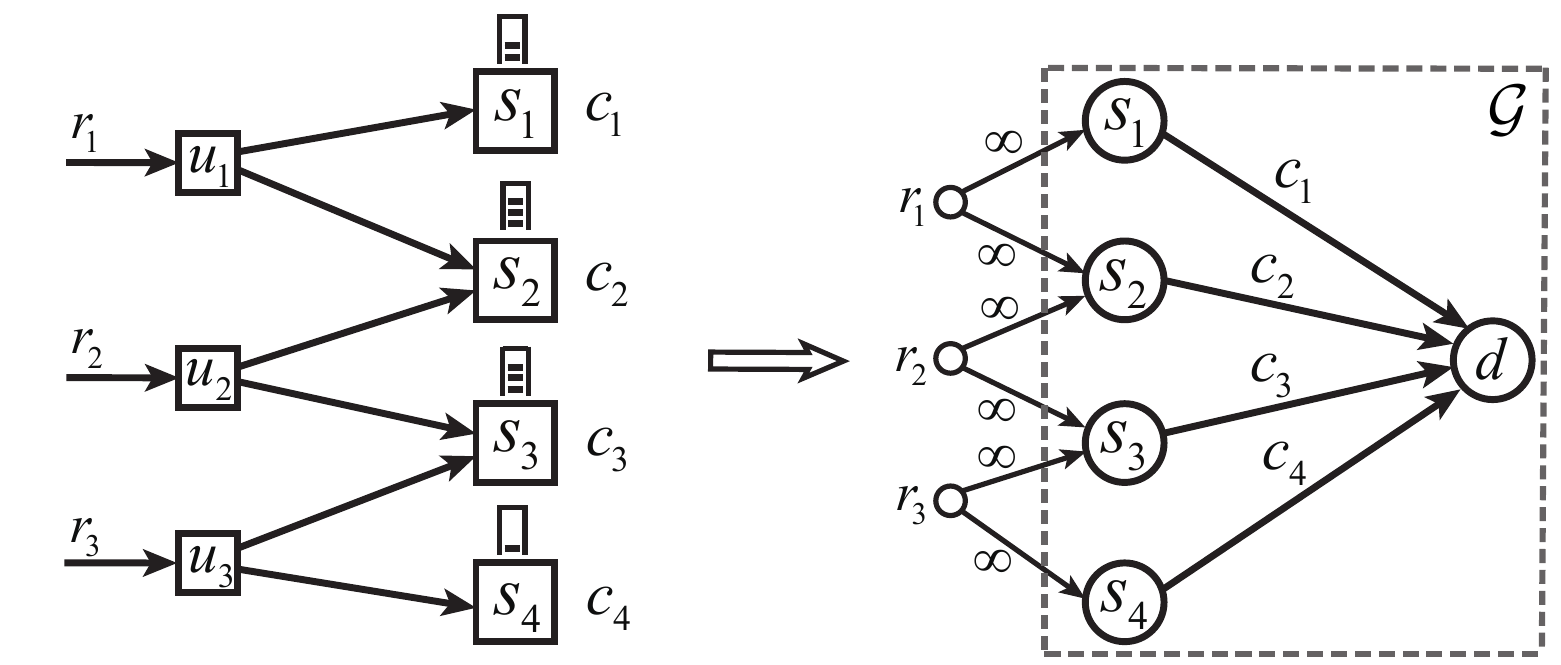}
	\vspace{-3mm}
	\caption{Correspondence between job scheduling and the L-NUM framework.}
	\vspace{-3mm}
	\label{fig:jobscheduling}
\end{figure}

We apply the L-NUM framework to the job scheduling application by creating a source node for each job classes, an intermediate node corresponding to each server and a virtual destination node (See Figure \ref{fig:jobscheduling}). The offered transmission rates of the links between server node and the virtual destination is equal to the time-varying capacity $c_m(t)$ of the servers, and the offered transmission rate between source nodes and intermediate nodes are infinity. 
 The job size $r_k(t)$ corresponds to the resource requirement of class $k$ jobs sent at $t$. 
Based on this correspondence, the P-GSMW policy achieves $\tilde{O}(T^{3/4})$-regret. Note that the max-weight scheduling component of the P-GSMW is equivalent to the Join-the-Shortest-Queue policy.

\subsection{Video Streaming}
In this example, we consider a network shared by $K$ users streaming video from $K$ corresponding servers. At each time slot, each server sends a chunk of the video file through the network to its corresponding user. The network operator determine the size of the chunks, which correspond to the rates of the video streams. It also controls the routing and scheduling in the network. User $k$ has a utility function $f_k$ that is unknown to the network operator, and obtains utility of value $f_k(r_k)$ after receiving a video chunk of size $r_k$. Here, we seek a policy that jointly adapts the video rates, i.e., determines the size of the video chunks and the routing and scheduling of the network such that the total utility obtained from the delivered video chunks is maximized. 

It is natural to map the L-NUM framework to the video streaming application. The network shared by the users plays the role of the network $\mathcal{G}$ in the L-NUM framework. Each user represents a traffic class. Each traffic class has the user's corresponding video server as the source node with the user node being the destination. The network states capture the possible time-variability in the network links (e.g. in wireless networks). The network action encapsulates the routing and scheduling actions of the network. The feasible action set $\mathcal{X}$ can captures constraints on network operations such as interference constraints and capacity constraints. Applying the P-GSMW policy, we obtain a joint rate-adaptation and network scheduling/routing policy with $\tilde{O}(T^{3/4})$-regret. The network action component here resembles the back-pressure algorithm.

\section{Simulations}\label{sec:simulations}

In this section, we evaluate the empirical performance of the P-GSMW policy under the L-NUM framework. We will also compare P-GSMW policy with GSMW algorithm (in an imaginary no-delay setting) to see the impact of feedback delay on the problem.

We instantiate the L-NUM framework on the job scheduling application. The example we construct for the simulation has 50 job schedulers (corresponding to 50 job classes) and 100 parallel servers. The links between job schedulers and servers are randomly generated with each scheduler having expected degree 6 (i.e., connected to 6 servers). The service rate of each server is generated by a uniform random variable with range $[0.5,1.5]$. We assign an underlying utility function to each class chosen from the four types: $f_k(r)=a_kr$ (linear function), $f_k(r)=a_k\sqrt{r+b_k}-a_k\sqrt{b_k}$ (square root function), $f_k(r)=-a_kr^2+b_kr$ (quadratic function), $f_k(r)=a_k\log(b_kr+1)$ (logarithmic function).

Applying the L-NUM framework to the example, we first form the corresponding optimization problem $\mathcal{P}$ and obtain that the optimal value $OPT(\mathcal{P})$ is equal to 84.4. We next run the P-GSMW policy and also the GSMW algorithm (\textbf{Algorithm \ref{alg:policy}}). Note that for the GSMW algorithm, we assume an imaginary no-delay setting where the utility values are immediately observable after decisions.

 We first investigate the effects of the parameter values $(\alpha,V,\delta)$ on the performance of the policy, then compare P-GSMW and GMSW, and finally study the impact of observation noise. 

\subsection{Choice of Parameter Values}
We vary the values of the parameters $(\alpha, V, \delta)$ in the GSMW policy and demonstrate their effects of the policy. The time horizon $T$ is set to 60000. When changing one parameter, the others are held fixed $(\alpha=5000, V=200, \delta = 0.005)$. We plot the queue length, defined as the sum of queue length at each server, and the instantaneous utility, defined as $\sum_kf_k(r_k(t))$ as the time evolves. The results on queue length are shown in {Figures  \ref{fig:parameter_value}, while the figures on instantaneous utility are much less readible nor informative, and is thus deferred to Appendix \ref{app:simulationfigure}. 

\textit{Parameter $\alpha$:}
 The parameter $\alpha$ essentially controls the step size of the P-GSMW policy with a larger $\alpha$ indicating a smaller step size. We vary $\alpha$ in $\{500, 1000, 5000, 10000 \}$. From the results, the average queue length decreases with the increase of $\alpha$. The queue length of a larger $\alpha$ tends to have larger and more persistent oscillation. Recalling the update of job sizes of the P-GSMW policy (Line \ref{alg:update_p}), such behavior can be attributed to that a larger $\alpha$ leads to a smaller ``negative feedback'' that the queue length has on job sizes.

\textit{Parameter $V$:} The parameter $V$ adjusts the relative weights of the P-GSMW policy on utility maximization and queue stability, with a larger $V$ indicating that the policy tries to increase the job sizes (and thus the instantaneous utility) more aggressively.  We vary $V$ in $\{ 50, 100, 200, 400\}$. Such behavior is clearly reflected in Figure \ref{fig:queue_V_1} as a larger $V$ leads to a larger steady-state queue size, but the difference is more obscure in the plot of instantaneous utility (figure omitted due to space constraint). We further calculate the time-average instantaneous utility of the P-GSMW policy under different values of $V$. Corresponding to $V=50, 100, 200, 400$, the time-average instantaneous utility are $84.8, 85.3, 86.1 ,87.1$, respectively. Note that the instantaneous utility can be larger than $OPT(\mathcal{P})$ since the virtual job size variables may not satisfy the capacity constraint of $\mathcal{P}$. The result further  supports that a larger $V$ leads to more aggressive increase in the job sizes.

\textit{Parameter $\delta:$} The parameter $\delta$ controls the approximation error of our estimate gradients with respect to the true gradients.  We vary $\delta$ in $\{0.005, 0.01, 0.05, 0.1 \}$. Due to that the underlying utility functions in our example do not have large curvature, the value of $\delta$ does not have significant effect on the policy.

%

\begin{figure*}[!tb]
	\centering
	\subfigure[]{
		\begin{minipage}[]{0.31\linewidth}
			\centering
			\vspace{-2mm}
			\includegraphics[width=1.0\linewidth]{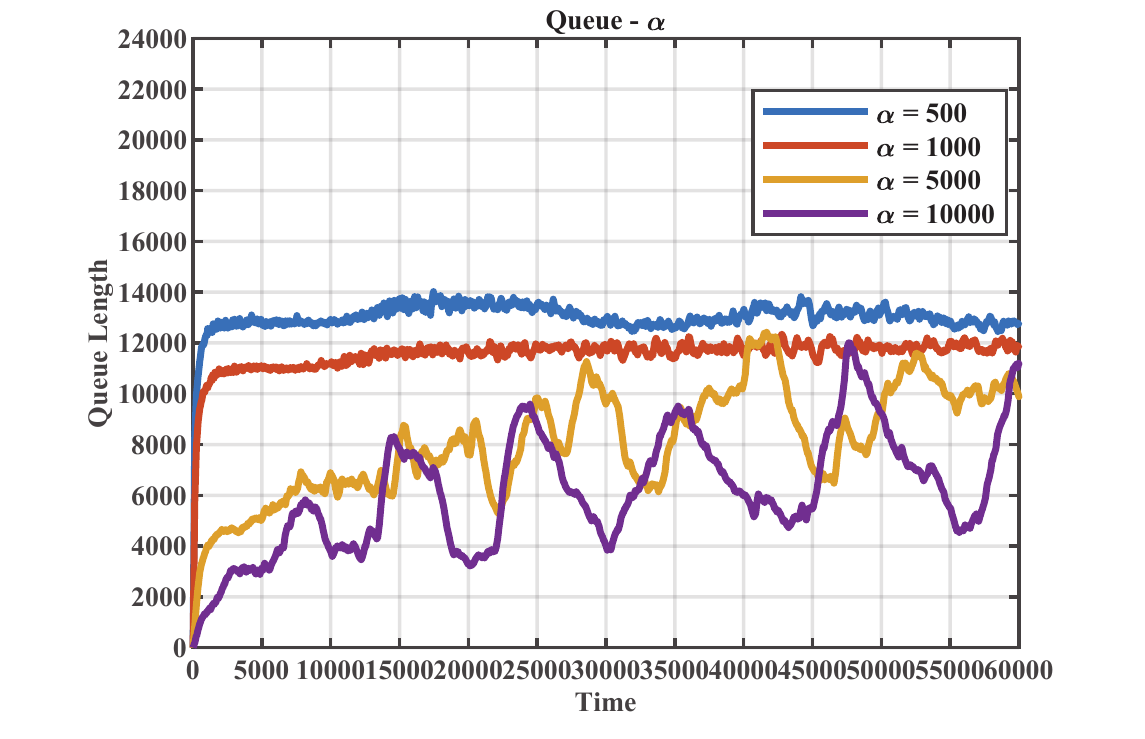}
			\vspace{-3mm}
			\label{fig:queue_alpha_1}
		\end{minipage}%
		
	}
	\subfigure[]{
		\begin{minipage}[]{0.31\linewidth}
			\centering
			\vspace{-2mm}
			\includegraphics[width=1.0\linewidth]{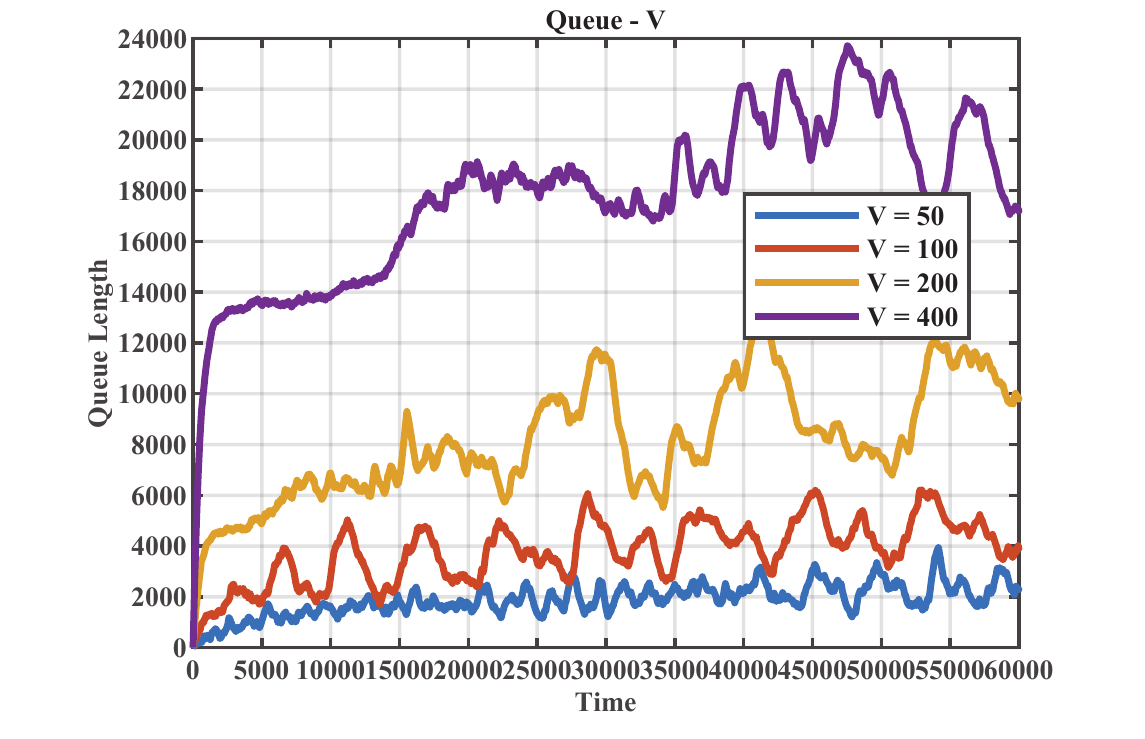}
			\vspace{-3mm}
			\label{fig:queue_V_1}
		\end{minipage}%
		
	}
	\subfigure[]{
		\begin{minipage}[]{0.31\linewidth}
			\centering
			\vspace{-2mm}
			\includegraphics[width=1.0\linewidth]{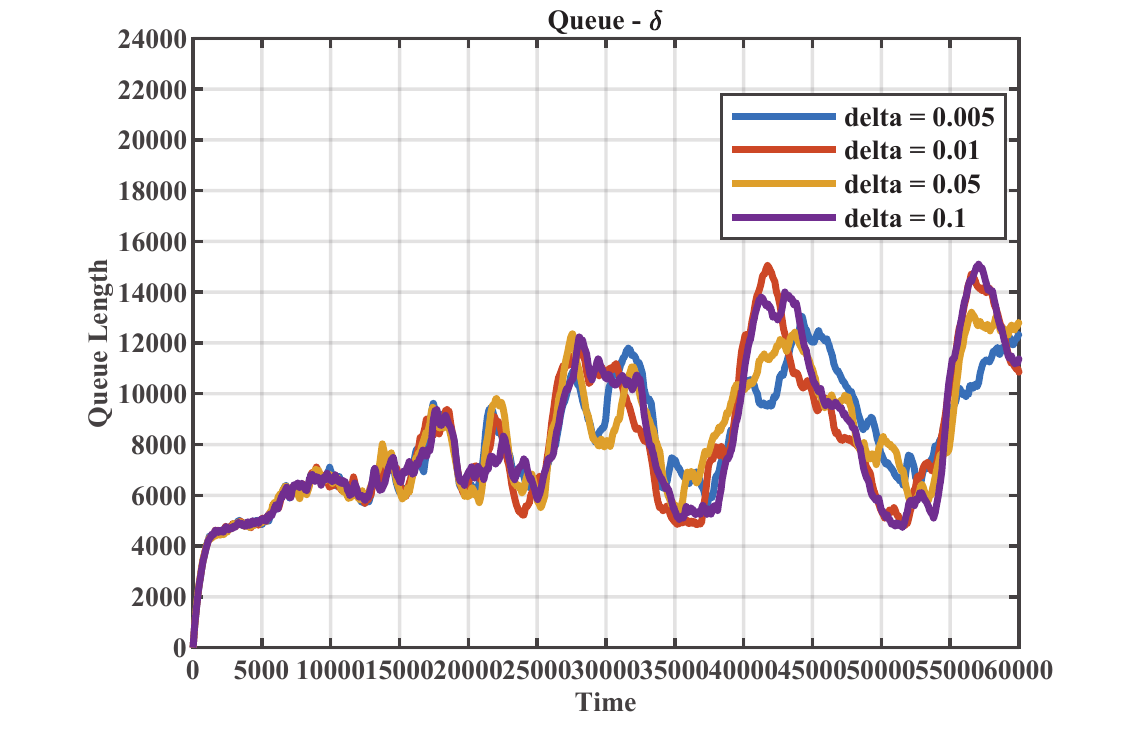}
			\vspace{-3mm}
			\label{fig:queue_delta_1}
		\end{minipage}%
		
	}
	\vspace{-3mm}
	\caption{The queue length under the P-GSMW policy with different parameter values.}
	\vspace{-3mm}
	\label{fig:parameter_value}		
\end{figure*}

\subsection{P-GSMW vs. GSMW}
We compare the behaviors of total queue length and instantaneous utility under P-GSMW and GSMW with the same parameter values of $(\alpha=5000, V=200, \delta = 0.005)$. Note that P-GSMW is run in our original setting (with queueing-style feedback delay) while GSMW is run in an imaginary no-delay setting.
It can be seen that, as in terms of job-size variables, P-GSMW switches between different instances of GSMW algorithms, both the queue length trajectory and the instantaneous utility trajectory under P-GSMW exhibits larger oscillation compared to those of GSMW.

\begin{figure}
	\centering
	\includegraphics[width=0.9\linewidth]{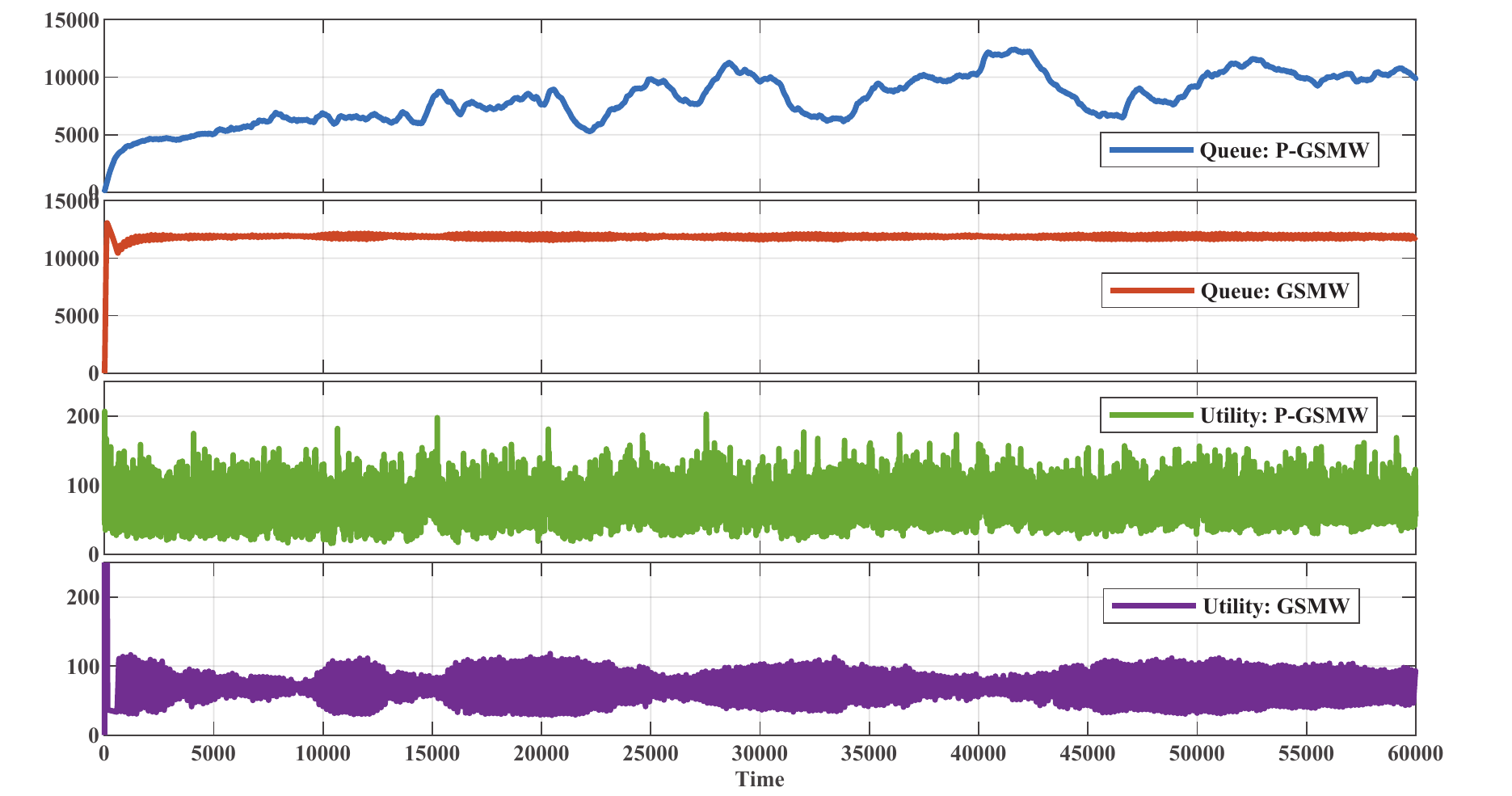}
	\vspace{-2mm}
	\caption{Queue length and instantaneous utility behavior under P-GSMW and GSMW.}
	\vspace{-3mm}
	\label{fig:comparison1}
\end{figure}

Furthermore, varying the time horizon $T$ in $\{10000, 20000, \ldots, 100000\}$ and setting $\alpha = 50\sqrt{T}, V=T^{1/4}, \delta = 1/\sqrt{T}$, we compare how the regret of P-GSMW and GSMW scales with the time horizon. Since it is computationally infeasible to compute the optimal strategy, we use $T$ times $OPT(\mathcal{P})$ as an upper bound of the expected utility achieved by the optimal strategy (see Theorem \ref{thm:upperbound}) and bound the regret by $T\cdot OPT(\mathcal{P})$ minus the utility achieved by the policies. We can see from Figure \ref{fig:compareison_regret1} that the regret of GSMW is lower by that of P-GSMW, which suggests that the feedback delay hurt the performance of the policy.

\begin{figure}
	\centering
	\includegraphics[width=0.8\linewidth]{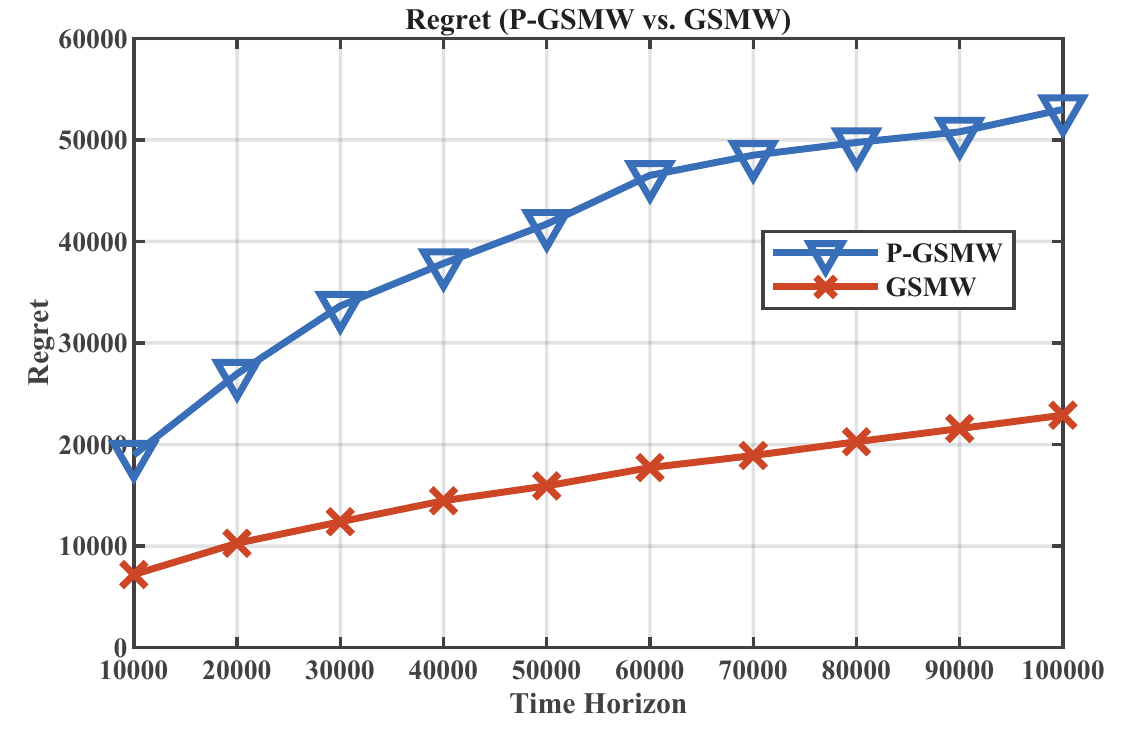}
	\vspace{-3mm}
	\caption{Regrets of P-GSMW and GSMW.}
	\vspace{-3mm}
	\label{fig:compareison_regret1}
\end{figure}
\subsection{Observation Noise}
We explore the situation where the utility observations are corrupted with noise and study the robustness of P-GSMW against such noise.
 We change the noise level from 0 (no noise) to 0.2 (each observation is corrupted with noise that is uniformly distributed in $[-0.2,0.2]$). 
 Varying the time horizon in $\{10000, 20000,\ldots, 100000 \}$ and setting $\alpha = 50\sqrt{T}, V=T^{1/4}, \delta = 1/\sqrt{T}$, we evaluate the scaling of regret under different noise levels.
 The results are plotted in Figures \ref{fig:regret1}.
 
 From Figures \ref{fig:regret1}, we see that the regrets of P-GSMW sublinear growth with time horizon even under a noise level of 0.2. Generally, the regret increases with noise level, but the difference is not significant for noise under 0.05.\footnote{To put this into perspective, there are 50 job classes and $OPT(\mathcal{P})$ is 84.4. The magnitude of the noise is about $0.05\times 50/84.4\simeq 3\%$ of the time-average utility.} The degradation of regret performance with noise can be attributed to that the variance of the gradient estimate.
 Recall that the gradient estimates of the P-GSMW policy at a time $t$ for class $k$ is equal to $\hat{\nabla}f_k(\hat{r}^I_k(t)):=\frac{{f}_k(\hat{r}^I_k(t)+\delta)-{f}_k(\hat{r}^I_k(t)-\delta)}{2\delta}$. If the two observations are corrupted by random noise $\epsilon_1, \epsilon_2$ respectively, then we have $\mathbb{E}[|\hat{\nabla}f_k(\hat{r}^I_k(t))|]\simeq \mathbb{E}[|\frac{{f}_k(\hat{r}^I_k(t)+\delta)-{f}_k(\hat{r}^I_k(t)-\delta)}{2\delta}|+\frac{|\epsilon_1-\epsilon_2|}{2\delta}]$. Due to the Lipschitz-continuity of $f_k$, the first term is of order $O(1)$ but the second is of $O(1/\delta)$. Thus, the second term dominates the magnitude of the gradient estimate and it increases with the magnitude of $\epsilon$ (i.e., the noise level). A gradient estimate of larger magnitude may lead to less stable updates, larger queue lengths, longer feedback delay, and ultimately, larger regret.

\begin{figure}
	\centering
	\includegraphics[width=0.8\linewidth]{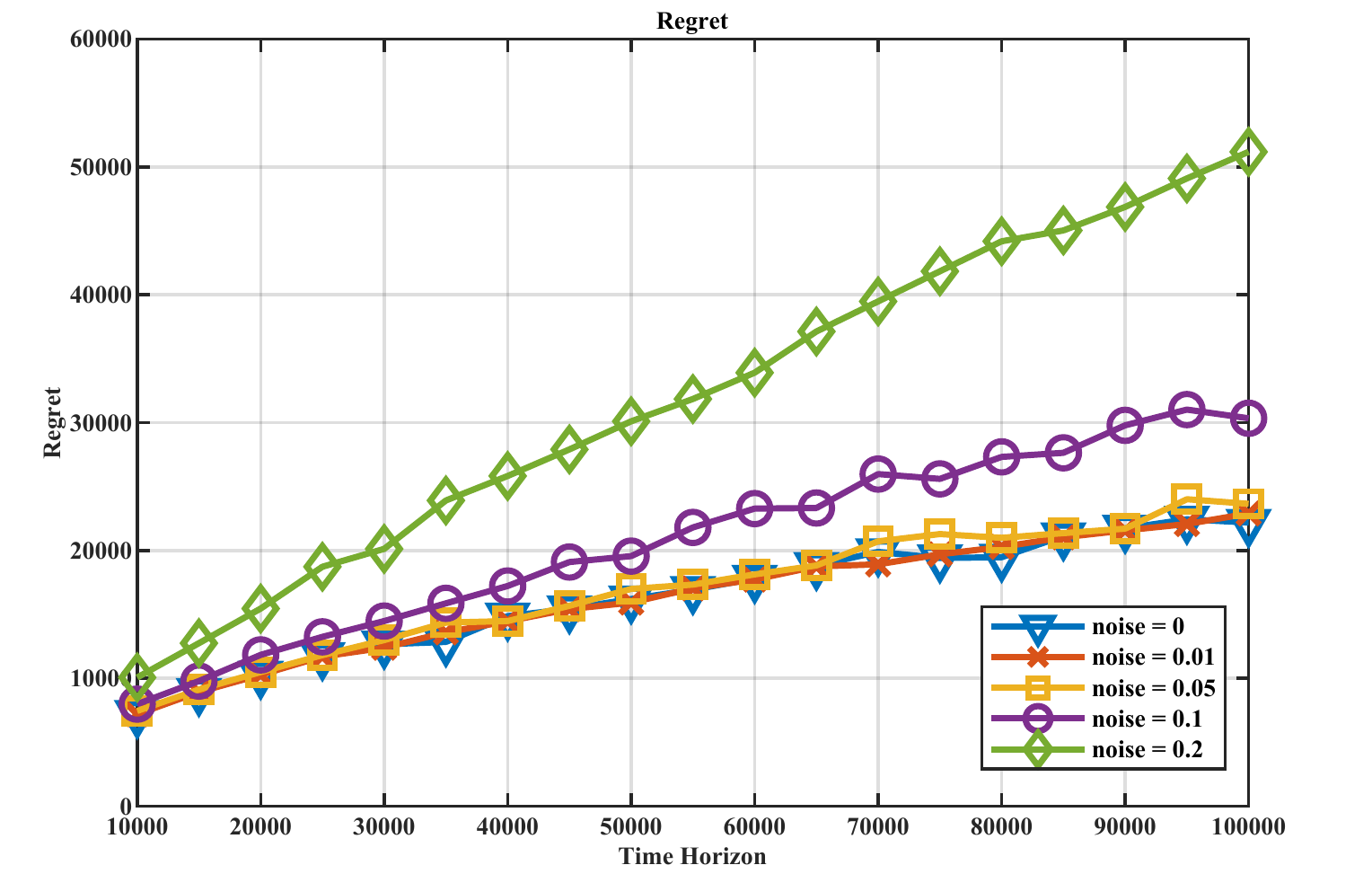}
	\vspace{-3mm}
	\caption{Regret of P-GSMW under different noise levels.}
	\vspace{-3mm}
	\label{fig:regret1}
\end{figure}

\section{Conclusion and Future Directions}\label{sec:conclusion}
 In this paper, we proposed a new NUM framework, Learning-NUM, where the utility functions are only accessible through zeroth-order feedback and the feedback experiences queueing-style delay. We upper bounded the expected utility achieved by any dynamic policy by the solution to a static optimization problem and designed an online scheduling policy (P-GSMW) that achieves sub-linear regret.

 Our scheduling policy achieves a regret of order $\tilde{O}(T^{3/4})$. This is worse than the existing lower bound of ${\Omega}(\sqrt{T})$, which was shown in the no-delay case. Hence, an important future direction is to determine whether the queueing-style delay of L-NUM fundamental increases the difficulty of the problem, i.e., a lower bound better than $\Omega(\sqrt{T})$ can be established, or that algorithm for L-NUM that has regret better than $\tilde{O}({T}^{3/4})$ exists. Finally, we have not theoretically studied the performance of P-GSMW policy under observation noise. Although it is expected that P-GSMW would have regret worse than $\tilde{O}(T^{3/4})$, how to minimize the adverse impact of the noise on the policy, and are there other methods that are more robust to noise are both questions of future interests.



\appendix
\section{Proofs of Theorems and Lemmas}\label{app:policyproof}
In this section, we present the proof of Theorem \ref{thm:regret}. For simplicity of notation, we will use $\hat{r}^I_k(t)$ or $\hat{r}_k(t)$ to denote the virtual job-size variable used at time $t$ (which suppresses the dependence of the invoked instance at $t$ on the time $t$).
Before doing so, we first lay out some preliminary results that will be useful in the subsequent analysis.

To begin with, we reiterate and prove the lemma that decomposes the regret into utility regret and queueing regret.
\begin{lemma}\label{lemma:decompose1}
	Let $\{r^*\}_k$ be the optimal solution to $\mathcal{P}$,
	\begin{align*}
	&R(\pi_{P-GSMW},T) \\\le & 2\mathbb{E}\left[\sum_{t=1}^{T}\sum_{k=1}^Kf_k(r^*_k)-f_k(\hat{r}_k(t)) \right]+C\sum_{i\in V}\sum_{k=1}^K\mathbb{E}[Q_i^k(T)]+CT\delta. 
	\end{align*}
\end{lemma}

\begin{proof}
First, observe that the utility obtained by $\pi_{P-GSMW}$ over the time horizon $T$ is equal to the total utility of the jobs sent from the sources minus the utility of the jobs that are not delivered at $T$. Recall that at time $t$, the two jobs sent from source $s_k$ have size $\hat{r}_k(t)+\delta$ and $\hat{r}_k(t)-\delta$, respectively. Since the utility of a single job is bounded, we have 
\begin{align*}
&\mathbb{E}[U(\pi_{P-GSMW},T)] \\\ge & \mathbb{E}\left[\sum_{t=1}^T\sum_{k=1}^K  {f}_k(\hat{r}_k(t)+\delta)+{f}_k(\hat{r}_k(t)-\delta)\right]-C\sum_{i\in V}\sum_{k=1}^K\mathbb{E}[Q_i^k(T)].
\end{align*}
By property (2) (Lipschitz continuity) of the underlying utility functions, we have 
\begin{align*}
&\mathbb{E}\left[\sum_{t=1}^T\sum_{k=1}^K  {f}_k(\hat{r}_k(t)+\delta)+{f}_k(\hat{r}_k(t)-\delta)\right]\\
\ge&2\mathbb{E}\left[\sum_{t=1}^{T}\sum_{k=1}^K {f}_k(\hat{r}_k(t))\right]-C\cdot T\delta,
\end{align*}
where the last inequality follows from
property (2) (Lipschitz continuity) of the underlying utility functions. 

By Theorem \ref{thm:upperbound} and that we are now assuming there are two injected jobs of each class at each time slot, $\sup_{\pi^*\in \bar{\Pi}}\mathbb{E}[U(\pi^*,T)]\le 2T\cdot OPT(\mathcal{P})=2\sum_{t=1}^T\sum_{k=1}^Kf_k(r^*_k)$. Putting the above analysis together, we obtain that
\begin{align*}
&R(\pi_{P-GSMW},T) \\\le & 2\mathbb{E}\left[\sum_{t=1}^{T}\sum_{k=1}^Kf_k(r^*_k)-f_k(\hat{r}_k(t)) \right]+C\sum_{i\in V}\sum_{k=1}^K\mathbb{E}[Q_i^k(T)]+CT\delta. 
\end{align*}
\end{proof}
 
Next, we show that the magnitude of the gradient estimates is bounded.
\begin{lemma}\label{lemma:gradbound}
	For all $k,t$, $\hat{\nabla}f_k(\hat{r}^I_k(t))\le L$ with probability 1.
\end{lemma}
\begin{proof}
	The lemma follows straightforwardly from the Lipschitz continuity of the underlying utility functions.
\end{proof}

We next show that the gradient estimate is unbiased with respect to a smoothed version of $f_k$, which is defined as $\tilde{f}_k(r)=\frac{1}{2\delta}\int_{-\delta}^{\delta}f_k(r+z)\mathrm{d}z$. Note that by definition $\tilde{f}_k$ is also concave and Lipschitz-continuous. Moreover, by the concavity and Lipschitz continuity of ${f}_k$, for all $r\in[\delta,B-\delta]$, $f_k(r)-C\delta\le \tilde{f}_k(r) \le f_k(r)$ \cite{cite:flaxman}.
\begin{lemma}\label{lemma:meangrad}
	For all $k,t$, $\hat{\nabla}f_k(\hat{r}_k(t))=\nabla\tilde{f}_k(\hat{r}_k(t))$.
\end{lemma}
\begin{proof}
	The lemma follows from the Fundamental Theorem of Calculus.
\end{proof}

Finally, we establish three basic properties of the updates of the P-GSMW policy. The first involves the update of virtual job size variables, the second considers the Max-Weight rule of choosing actions and the third deals with the queue dynamics. 
\begin{lemma}\label{lemma:update}
	For each $t$, let $I$ be the instance invoked at time $t$, we have for any $\{r \}_k$ with $r_k\in[\delta, B-\delta]$
	\begin{align*}
	&\sum_{k=1}^K \left[V\hat{\nabla}f_k(\hat{r}^I_k(t-1))(r_k-\hat{r}^I_k(t-1))\right] +\sum_{k=1}^K\left[Q^{k}_{s_{k}}(t)\hat{r}^I_k(t)  \right]\\
	\le & \sum_{k=1}^K\left[Q^{k}_{s_{k}}(t)r_k+\alpha[(\hat{r}^I_k(t-1)-r_k)^2-(\hat{r}^I_k(t)-r_k)^2]+{C}.   \right]
	\end{align*}
\end{lemma}
\begin{proof}
	
	From Line \ref{alg:update_p} of \textbf{Algorithm \ref{alg:p-gsmw}}, since the projection operator is a contraction, we have for each $k$
	\begin{align*}
	&(\hat{r}^I_k(t)-r_k)^2\\
	\le & \left[ \hat{r}^I_k(t-1)+\frac{1}{\alpha} (V\cdot \hat{\nabla}f_k(\hat{r}_k(t-1))-Q^{k}_{s_{k}}(t)) -r_k \right]^2\\
	=&  (\hat{r}^I_k(t-1)-r_k)^2 + \frac{1}{\alpha}[\hat{\nabla}f_k(\hat{r}^I_k(t-1))(\hat{r}_k(t-1)-r_k)-Q^{k}_{s_{k}}(t)(\hat{r}^i_k(t-1)-r_k)]\\
	&\quad +\frac{V^2(\hat{\nabla}f_k(\hat{r}^I_k(t-1))^2}{\alpha^2}-\frac{2V\cdot \hat{\nabla}f_k(\hat{r}^I_k(t-1))Q^{k}_{s_{k}}(t)}{\alpha^2}+\frac{Q^{k}_{s_{k}}(t)^2}{\alpha^2}
	\end{align*}
	Since $\alpha = 2K\sqrt{T}/\eta,V=T^{1/4}$ for all $k,\tau$, we have $\frac{V^2(\hat{\nabla}f_k(\hat{r}^I_k(t-1))^2}{\alpha^2}\le \frac{C}{\sqrt{T}}$. Plugging these in and rearranging the term, we obtain
	\begin{align*}
	&V\hat{\nabla}f_k(\hat{r}^I_k(t-1))(r_k-\hat{r}^I_k(t-1)) +Q^{k}_{s_{k}}(t)\hat{r}^I_k(t) +Q^{k}_{s_{k}}(t)[\hat{r}^I_k(t-1)-\hat{r}^I_k(t)]  \\
=&V\hat{\nabla}f_k(\hat{r}^I_k(t-1))(r_k-\hat{r}^I_k(t-1)) +Q^{k}_{s_{k}}(t)\hat{r}^I_k(t-1)  \\
\le & Q^{k}_{s_{k}}(t)r_k+\alpha[(\hat{r}^I_k(t-1)-r_k)^2-(\hat{r}^I_k(t)-r_k)^2]+\frac{Q^k_{s_{k}}(t)^2}{\alpha}\\ &\ - \frac{2V\cdot \hat{\nabla}f_k(\hat{r}^I_k(t-1))Q^{k}_{s_{k}}(t)}{\alpha^2}+{C}.   \\
	\end{align*}
	As $\delta=T^{-1/2}$, the lemma trivially hold for $\hat{r}^I_k(t)=\delta$ for all $k$. Hence, we only need to consider the case where $\hat{r}^I_k(t)>\delta$, which implies that $\hat{r}^I_k(t)-\hat{r}^I_k(t-1)\le \frac{1}{\alpha} (V\cdot \hat{\nabla}f_k(\hat{r}_k(t-1))-Q^{k}_{s_{k}}(t))$. It follows that
	\begin{align*}
	&V\hat{\nabla}f_k(\hat{r}^I_k(t-1))(r_k-\hat{r}^I_k(t-1)) +Q^{k}_{s_{k}}(t)\hat{r}^I_k(t)  \\
	\le & Q^{k}_{s_{k}}(t)r_k+\alpha[(\hat{r}^I_k(t-1)-r_k)^2-(\hat{r}^I_k(t)-r_k)^2] - \frac{V\cdot \hat{\nabla}f_k(\hat{r}^I_k(t-1))Q^{k}_{s_{k}}(t)}{\alpha^2}+{C}.   \\
	\end{align*}
	 By Lemma \ref{lemma:meangrad}, since $\tilde{f}$ is non-decreasing, $\hat{\nabla}f_k(\hat{r}^I_k(t-1))=\nabla{\tilde{f}_k(\hat{r}^I_k(t-1))}\ge 0$, we have $\frac{2V\cdot \hat{\nabla}f_k(\hat{r}^I_k(t-1))Q^{k}_{s_{k}}(t)}{\alpha^2}\ge 0$. Therefore,
	 \begin{align*}
	 &V\hat{\nabla}f_k(\hat{r}^I_k(t-1))(r_k-\hat{r}^I_k(t-1)) +Q^{k}_{s_{k}}(t)\hat{r}^I_k(t)  \\
	 \le & Q^{k}_{s_{k}}(t)r_k+\alpha[(\hat{r}^I_k(t-1)-r_k)^2-(\hat{r}^I_k(t)-r_k)^2] +{C}.   \\
	 \end{align*}
	The lemma follows by summing over $k$.
\end{proof}

\begin{lemma}\label{lemma:maxweight}
	At every even time slot $t$, for any $\bm{x}\in \mathcal{X}$, and for all $\{r\}_k$,
	\begin{align*}
	&\sum_{k=1}^KQ^{k}_{s_{k}}(t)\left[{r}_k-\sum_{j\in\mathcal{N}_{s_k}}A_{s_kj}^k(\omega(t),{\bm{x}}(t))\right]\\&\quad+\sum_{k=1}^K\sum_{i\in V,i\neq s_k}Q_i^k(t)\left[\sum_{j:i\in \mathcal{N}_j}{A}_{ji}^k(\omega(t),{\bm{x}}(t))-\sum_{j\in \mathcal{N}_i}{A}_{ij}^k(\omega(t),{\bm{x}}(t))\right]
	\\\le & \sum_{k=1}^KQ^{k}_{s_{k}}(t)\left[{r}_k-\sum_{j\in\mathcal{N}_{s_k}}A_{s_kj}^k(\omega(t),{\bm{x}})\right]\\&\quad+\sum_{k=1}^K\sum_{i\in V,i\neq s_k}Q_i^k(t)\left[\sum_{j:i\in \mathcal{N}_j}{A}_{ji}^k(\omega(t),{\bm{x}})-\sum_{j\in \mathcal{N}_i}{A}_{ij}^k(\omega(t),{\bm{x}})\right]
	\end{align*}
\end{lemma}
\begin{proof}
	By rearranging the terms, we recover exactly the right-hand-side of Line \ref{alg:maxweight_p} of \textbf{Algorithm \ref{alg:p-gsmw}}.
	The lemma then follows from the construction of the Max-Weight update rule. Note that the inequality holds for all $\{r\}_k$ since the terms involving $\{r\}_k$ do not affect the maximization.
\end{proof}

\begin{lemma}\label{lemma:dynamics}
	For each $k,t$, recall that $\hat{r}_k(t)=\hat{r}_k^i(t)$ for the invoked instance $i$. 
	\begin{align}
	&Q_{s_k}^k(t+1)^2-Q_{s_k}^k(t)^2\le & 4Q_{s_k}^k(t)\left[\hat{r}_k(t)-\sum_{j\in\mathcal{N}_{s_k}}A_{s_kj}^k(\omega(t),{\bm{x}}(t))\right] + C, \label{ieq:driftd1}
	\end{align}
	and for each $i\in V,k, i\neq s_k,d_k$,
	\begin{align}
	&Q_{i}^k(t+1)^2-Q_{i}^k(t)^2\nonumber\\
	\le& 4Q_i^k(t)\left[\sum_{j:i\in \mathcal{N}_j}{A}_{ji}^k(\omega(t),{{\bm{x}}(t)})-\sum_{j\in \mathcal{N}_i}{A}_{ij}^k(\omega(t),{{\bm{x}}(t)})\right] + C. \label{ieq:driftd2}
	\end{align}
\end{lemma}
\begin{proof}
	Note that each time slot now correspond to two time slots in our original model of Section \ref{sec:model}. The lemma follows directly from the dynamics of the queue evolution. For (\ref{ieq:driftd1}), we have
	\begin{align}
	&Q_{s_k}^k(t+1)^2-Q_{s_k}^k(t)^2\nonumber\\\le& \left[Q_{s_k}^k(t)+2\hat{r}_k(t)-2\sum_{j\in\mathcal{N}_{s_k}}A_{s_kj}^k(\omega(t),{\bm{x}}(t))\right]^2 -Q_{s_k}^k(t)^2\nonumber\\
	=&Q_{s_k}^k(t)^2+4Q_{s_k}^k(t)\left[\sum_{k\in n}\hat{r}_k(t)-\sum_{j\in\mathcal{N}_{s_k}}A_{s_kj}^k(\omega(t),\hat{\bm{x}}(t))\right]\nonumber\\ &\quad + 4\left[\hat{r}_k(t)-\sum_{j\in\mathcal{N}_{s_k}}A_{s_kj}^k(\omega(t),{\bm{x}}(t))\right]^2-Q_{s_k}^n(t)^2\nonumber\\
	\le& 4Q_{s_k}^k(t)\left[\hat{r}_k(t)-\sum_{j\in\mathcal{N}_{s_k}}A_{s_kj}^k(\omega(t),\hat{\bm{x}}(t))\right] + C. \nonumber
	\end{align}
	Inequality (\ref{ieq:driftd2}) follows similarly.
\end{proof}

\subsubsection{Queueing Regret}
In this section, we bound the queueing regret by providing a bound on the expected queue size at $T$.
To do so, we will first use $\sum_{i\in V}\sum_{k=1}^KQ_i^k(t)^2$ as a Lyapunov function and prove that the Lyapunov function has expected conditional negative drift, which combined with a result on discrete stochastic process from \cite{cite:neely1}, leads to a bound on $\sum_{i\in V}\sum_{k=1}^KQ_i^k(t)$ both in expectation and with high probability.

Define $\bm{Q}(t)$ to be the vector that includes all the queues $\{Q(t)\}_i^k$ as coordinates  and $||\cdot||$ as the Euclidean norm. By Lemma \ref{lemma:dynamics}, we have
%
\begin{align}
&||\bm{Q}(t+1)||^2-||\bm{Q}(t)||^2 \nonumber\\ =& 
\sum_{k=1}^KQ_{s_k}^k(t+1)^2-Q_{s_k}^k(t)^2 + \sum_{i\in V}\sum_{k=1}^K Q_i^k(t+1)^2-Q_i^k(t)^2 \\
\le &
4\sum_{k=1}^KQ_{s_k}^k(t)\left[\hat{r}_k(t)-\sum_{j\in\mathcal{N}_{s_k}}A_{s_kj}^k(\omega(t),{\bm{x}}(t))\right]\\&\quad+ 4\sum_{i\in V}\sum_{k=1}^KQ_k^n(t)\left[\sum_{j:i\in \mathcal{N}_j}{A}_{ji}^k(\omega(t),\bm{x}(t))-\sum_{j\in \mathcal{N}_i}{A}_{ij}^k(\omega(t),\bm{x}(t))\right]+ C. \label{ieq:squaredrift}
\end{align}
Next, we prove a conditional drift argument on $||\bm{Q}(t)||^2$ under the P-GSMW policy.
\begin{lemma}\label{lemma:drift}
	There exists $\epsilon>0$ such that under the P-GSMW policy,
	\begin{align*}
	\mathbb{E}[||\bm{Q}(t+1)||^2-||\bm{Q}(t)||^2\mid \bm{Q}(t)] \le -\epsilon  \sum_{i\in V}\sum_{k=1}^KQ_i^k(t) + C\sqrt{T}.
	\end{align*}
\end{lemma}
\begin{proof}
	Continuing from Lemma \ref{lemma:update}, rearranging terms, we have for any $\{r\}_k$ with $r_k\in[\delta,B-\delta]$,
	\begin{align}
	&\sum_{k=1}^KQ^k_{s_k}(t)\hat{r}_k(t)\nonumber\\
	\le &\sum_{k=1}^K\left[Q^{k}_{s_{k}}(t)r_k+\alpha[(\hat{r}^i_k(t-1)-r_k)^2-(\hat{r}^i_k(t)-r_k)^2]+{C}   \right]\nonumber\\
	&\quad+\sum_{k=1}^K V\hat{\nabla}f_k(\hat{r}_k(t-1))(r_k-\hat{r}_k(t-1))\nonumber\\
	\le& CV +C\alpha +\sum_{k=1}^K Q_{s_k}^k(t)r_k \le C\sqrt{T}+\sum_{k=1}^K Q_{s_k}^k(t)r_k,\label{ieq:lemmadrift}
	\end{align}
	where inequality (\ref{ieq:lemmadrift}) follows from that $\alpha =O(\sqrt{T}), V=O(T^{1/4})$. Adding same terms to both sides of (\ref{ieq:lemmadrift}), 
	\begin{align}
	&\sum_{k=1}^KQ_{s_k}^k(\tau)\left[\hat{r}_k(t)-\sum_{j\in\mathcal{N}_{s_k}}A_{s_kj}^k(\omega(t),{\bm{x}}(t))\right]\nonumber\\&\quad+ \sum_{i\in V}\sum_{k=1}^KQ_i^k(t)\left[\sum_{j:i\in \mathcal{N}_j}{A}_{ji}^k(\omega(t),\bm{x}(t))-\sum_{j\in \mathcal{N}_i}{A}_{ij}^k(\omega(t),\bm{x}(t))\right]
	\nonumber\\\le & \sum_{k=1}^KQ_{s_k}^k(t)\left[{r}_k-\sum_{j\in\mathcal{N}_{s_k}}A_{s_kj}^k(\omega(t),{\bm{x}}(t))\right]+C\sqrt{T}\nonumber\\&\quad+ \sum_{i\in V}\sum_{k=1}^KQ_i^k(t)\left[\sum_{j:i\in \mathcal{N}_j}{A}_{ji}^k(\omega(t),\bm{x}(t))-\sum_{j\in \mathcal{N}_i}{A}_{ij}^k(\omega(t),\bm{x}(t))\right]
	 \label{ieq:crossterm}
	\end{align}
	We take $\{r\}_k$ to be $r_k=\delta$. By Lemma \ref{lemma:maxweight}, we have for any $\bm{x}\in\mathcal{X}$,
	\begin{align}
	& \sum_{k=1}^KQ_{s_k}^k(t)\left[{r}_k-\sum_{j\in\mathcal{N}_{s_k}}A_{s_kj}^k(\omega(t),{\bm{x}}(t))\right]\nonumber\\&\quad+ \sum_{i\in V}\sum_{k=1}^KQ_i^k(t)\left[\sum_{j:i\in \mathcal{N}_j}{A}_{ji}^k(\omega(t),\bm{x}(t))-\sum_{j\in \mathcal{N}_i}{A}_{ij}^k(\omega(t),\bm{x}(t))\right]\nonumber\\
	\le & \sum_{k=1}^KQ_{s_k}^k(t)\left[\delta-\sum_{j\in\mathcal{N}_{s_k}}A_{s_kj}^k(\omega(t),{\bm{x}}(t))\right]\nonumber\\&\quad+ \sum_{i\in V}\sum_{k=1}^KQ_i^k(t)\left[\sum_{j:i\in \mathcal{N}_j}{A}_{ji}^k(\omega(t),\bm{x}(t))-\sum_{j\in \mathcal{N}_i}{A}_{ij}^k(\omega(t),\bm{x}(t))\right].\label{ieq:randomized}
	\end{align}
	
	Combining (\ref{ieq:squaredrift}), (\ref{ieq:crossterm}) and (\ref{ieq:randomized}), we obtain that for any $\bm{x}\in\mathcal{X}$
	\begin{align*}
	&||\bm{Q}(t+1)||^2-||\bm{Q}(t)||^2 \\ \le&
	4\sum_{k=1}^KQ_{s_k}^k(t)\left[\delta-\sum_{j\in\mathcal{N}_{s_k}}A_{s_kj}^k(\omega(t),{\bm{x}}(t))\right]+C\sqrt{T}\\&\quad+ 4\sum_{i\in V}\sum_{k=1}^KQ_i^k(t)\left[\sum_{j:i\in \mathcal{N}_j}{A}_{ji}^k(\omega(t),\bm{x}(t))-\sum_{j\in \mathcal{N}_i}{A}_{ij}^k(\omega(t),\bm{x}(t))\right]
	\end{align*}
	Since $\omega(t)$'s are i.i.d. $\omega(t)$ is independent of $\bm{Q}(t)$ which only depends on system information before $t$, we have for each fixed $\bm{x},$
	\begin{align}
	&	\mathbb{E}\left[  Q_{s_k}^k(t)\left[\delta-\sum_{j\in\mathcal{N}_{s_k}}A_{s_kj}^k(\omega(t),{\bm{x}})-{\eta}\right]\mid \bm{Q}(t)\right]\nonumber\\
	=& Q_{s_k}^k(t)\cdot \mathbb{E}\left[\delta-\sum_{j\in\mathcal{N}_{s_k}}A_{s_kj}^k(\omega(t),{\bm{x}})\right]\nonumber\\
	=& Q_{s_k}^k(t)\cdot \sum_{\omega\in\mathcal{W}}p(\omega) \left[\delta-\sum_{j\in\mathcal{N}_{s_k}}A_{s_kj}^k(\omega,{\bm{x}}) \right] \label{ieq:slater1}
	\end{align}
	Similarly,
	\begin{align}
	&\mathbb{E}\left[Q_i^k(t)\left[\sum_{j:i\in \mathcal{N}_j}{A}_{ji}^k(\omega(t),\bm{x}(t))-\sum_{j\in \mathcal{N}_i}{A}_{ij}^k(\omega(t),\bm{x}(t)) \right]\mid \bm{Q}(t)\right]\nonumber\\
	\le & Q_{i}^k(t)\cdot \sum_{\omega\in\mathcal{W}}p(\omega)\left[\sum_{j:i\in \mathcal{N}_j}{A}_{ji}^k(\omega,\bm{x})-\sum_{j\in \mathcal{N}_i}{A}_{ij}^k(\omega,\bm{x})\right] \label{ieq:slater2}
	\end{align}
	Let $\epsilon = \frac{\eta-\delta}{2}>0$. By the Slater's condition and that $\Lambda(\omega)$ is downward closing, combining (\ref{ieq:slater1}) and (\ref{ieq:slater2}), we have
	\begin{align}
	\mathbb{E}[||\bm{Q}(t+1)||^2-||\bm{Q}(t)||^2\mid \bm{Q}(t)] \le -\epsilon  \sum_{i\in V}\sum_{k=1}^KQ_i^k(t) + C\sqrt{T}\label{ieq:randomized1}
	\end{align}
\end{proof}

Lemma \ref{lemma:drift} establishes that $||\bm{Q}(t)||^2$ tends to decrease when the queue length is significantly larger than $O(\sqrt{T})$. To bound the queueing regret based on this, we use the following drift lemma for stochastic processes from \cite{cite:neely1}. We will not need the full generality of the lemma as it provides expectation and with-high-probability bound on stochastic processes that satisfy multi-slot drift condition, but we only need to deal with single-slot drift.
\begin{lemma}{\cite{cite:neely1}}\label{lemma:neely}
	Let $\{Z(t),t\ge 0\}$ be a discrete time stochastic process adapted to a filtration $\{\mathcal{F}(t),t\ge 0\}$ with $Z(0)=0$ and $\mathcal{F}(0)=\{\emptyset,\Omega\}$. Suppose there exists an integer $t_0>0$, real constants $\theta>0$, $\delta_{max}>0$ and $0<\xi\le \delta_{max}$ such that	
	\begin{align}
	|Z(t+1)-Z(t)|&\le \delta_{max}\\
	\mathbb{E}[Z(t+t_0)-Z(t)\mid\mathcal{F}(t)]&\le t_0\delta_{\max},\quad\mbox{if $Z(t)< \theta$}\\
	\mathbb{E}[Z(t+t_0)-Z(t)\mid\mathcal{F}(t)]&\le -t_0\zeta,\quad\mbox{if $Z(t)\ge \theta$}.
	\end{align}
	hold for all $t\in\{1,2,\ldots,\}$, then 
	\begin{align}
	\mathbb{E}[Z(t)]\le \theta + t_0\delta_{max}+t_0\frac{4\delta_{max}^2}{\xi}\log \frac{8\delta^2_{max}}{\xi^2}, \forall t\in\{1,2,\ldots,\}
	\end{align}
	and
	\begin{align}
	\forall\ 0<\mu <1, \mathbb{P}(Z(t)\ge z)\le \mu, \forall t\in \{1,2,\ldots,\},\label{ieq:highprobability}\end{align}
	where $z=\theta+t_0\delta_{max}+t_0\frac{4\delta_{max}^2}{\xi}\log\frac{8\delta^2_{max}}{\xi^2}+t_0\frac{4\delta^2_{max}}{\xi}\log\frac{1}{\mu}$. 
	
\end{lemma}
Continuing from Lemma \ref{lemma:drift}, since $\sum_{n\in V, k}Q_n^k(t)\ge ||\bm{Q}(t)||$ (as $l_1$ norm is no smaller than the Euclidean norm), we have
\begin{align*}
\mathbb{E}[||\bm{Q}(t+1)||^2-||\bm{Q}(t)||^2\mid \bm{Q}(t)] \le -\epsilon ||\bm{Q}(t)|| + C\sqrt{T}.
\end{align*}
It follows that 
\begin{align*}
\mathbb{E}[||\bm{Q}(t+1)||^2\mid \bm{Q}(t)] \le& ||\bm{Q}(t)||^2 -\epsilon  ||\bm{Q}(t)|| + C\sqrt{T}\\
\le & \left(||\bm{Q}(t)||-\epsilon\right)^2 \quad\mbox{when $||\bm{Q}(t)||>C\sqrt{T}$}.
\end{align*}
It follows that when $||\bm{Q}(t)||>C\sqrt{T}$,
\begin{align*}
\mathbb{E}[||\bm{Q}(t+1)||\mid \bm{Q}(t)] \le \sqrt{\mathbb{E}[||\bm{Q}(t+1)||^2\mid \bm{Q}(t)]} \le ||\bm{Q}(t)||-\epsilon.
\end{align*}
Further, since $||\bm{Q}(t+1)||-||\bm{Q}(t)||\le ||\bm{Q}(t+1)-\bm{Q}(t)||\le C$. Hence, invoking Lemma \ref{lemma:neely} with $t_0=1$, $\theta = C\sqrt{T}, \delta_{max} = C, \zeta = C$, we obtain that
$\mathbb{E}[||\bm{Q}(t)||]\le \tilde{O}(T^{1/2})$ for all $t$. By Cauchy-Schwarz inequality, $\mathbb{E}[\sum_{i\in V}\sum_{k=1}^KQ_i^k(t)]\le N|V|\cdot \mathbb{E}[||\bm{Q}(\tau)||]\le \tilde{O}(T^{1/2})$ for all $t$. Furthermore, by union bound, we also have that there exists a constant $C$ such that with probability at least $1-1/T$, $\sum_{i\in V}\sum_{k=1}^KQ_i^k(t)\le CT\log T$.

With the analysis above, we summarize the result on queueing regret  in the following theorem.
\begin{theorem}\label{thm:workload}
	Under P-GSMW, $\forall t=1,\ldots,T$, $\sum_{i\in V}\sum_{k=1}^KQ_i^k(t)\le \tilde{O}(\sqrt{T})$ in expectation and with high probability. In particular,\\ $\mathbb{E}[\sum_{i\in V}\sum_{k=1}^KQ_i^k(T)]\le \tilde{O}(\sqrt{T})$
\end{theorem}


\subsubsection{Utility Regret}

In this section, we bound the utility regret term. We first decompose the utility regret into three components as follow
\begin{align}
&\mathbb{E}\left[\sum_{t=1}^{T}\sum_{k=1}^Kf_k(r^*_k)-f_k(\hat{r}_k(t)) \right]\nonumber\\ 
=&\mathbb{E}\left[\sum_{t=1}^{T}\sum_{k=1}^Kf_k(r^*_k)-{f}_k(\hat{r}^*_k) \right]+ \mathbb{E}\left[\sum_{t=1}^{T}\sum_{k=1}^K{f}_k(\hat{r}^*_k)-\tilde{f}_k(\hat{r}^*_k) \right]\nonumber \\&\quad +\mathbb{E}\left[\sum_{t=1}^{T}\sum_{k=1}^K\tilde{f}_k(\hat{r}^*_k)-\tilde{f}_k(\hat{r}_k(t)) \right] + \mathbb{E}\left[\sum_{t=1}^{T}\sum_{k=1}^K\tilde{f}_k(\hat{r}_k(t))-f_k(\hat{r}_k(t)) \right], \label{ieq:utilitydecompose}
\end{align}
where $(\hat{r}_1^*,\ldots,\hat{r}_K^*)$ is the vector that maximizes $\sum_{k=1}^Kf_k(r_k)$ subject to $(r_1,\ldots,r_K)\in Cap(\mathcal{G})$ and $\forall k,\ r_k\in [\delta, B-\delta]$, i.e., the optimal solution to $\mathcal{P}$ restricting to each $r_k\in [\delta,B-\delta]$.
As $f_k$ is Lipschitz continuous, by Lemma \ref{lemma:aux1}, we have $\sum_{k=1}^Kf_k(r_k^*)-f_k(\hat{r}_k^*)\le C\delta.$ Further, $\sum_{k=1}^K{f}_k(\hat{r}_k^*)-\tilde{f}_k(r_k^*)\le KL\delta$. Since $f_k$ is concave, $\sum_{k=1}^K\tilde{f}_k(\hat{r}_k(t))-f_k(\hat{r}_k(t))\le 0$. It follows that
\begin{align*}
&\mathbb{E}\left[\sum_{t=1}^{T}\sum_{k=1}^Kf_k(r^*_k)-f_k(\hat{r}_k(t)) \right]
\le  \mathbb{E}\left[\sum_{t=1}^{T}\sum_{k=1}^K\tilde{f}_k(\hat{r}^*_k)-\tilde{f}_k(\hat{r}_k(t)) \right] +CT\delta\\
\le&  \mathbb{E}\left[\sum_{t=1}^{T}\sum_{k=1}^K\tilde{f}_k(\hat{r}^*_k)-\tilde{f}_k(\hat{r}_k(t)) \right] +C\sqrt{T}
\end{align*}
Hence, we can focus on bounding $\mathbb{E}\left[\sum_{t=1}^{T}\sum_{k=1}^K\tilde{f}_k(\hat{r}^*_k)-\tilde{f}_k(\hat{r}_k(t)) \right]$.

Again, starting from Lemma \ref{lemma:update} and plugging in $\{\hat{r}^*\}_k$, we have 
\begin{align*}
&\sum_{k=1}^K \left[V\hat{\nabla}f_k(\hat{r}^I_k(t-1))(\hat{r}^*_k-\hat{r}^I_k(t-1))\right] +\sum_{k=1}^K\left[Q^{k}_{s_{k}}(t)\hat{r}^I_k(t)  \right]\\
\le & \sum_{k=1}^K\left[Q^{k}_{s_{k}}(t)r_k+\alpha[(\hat{r}^I_k(t-1)-\hat{r}^*_k)^2-(\hat{r}^I_k(t)-\hat{r}^*_k)^2]+{C}.   \right]
\end{align*}
Again, multiplying both sides by two and adding the same terms on both sides lead to
\begin{align}
&\sum_{k=1}^K \left[V\hat{\nabla}f_k(\hat{r}^I_k(t-1))(\hat{r}^*_k-\hat{r}^I_k(t-1))\right]\\&\quad +\sum_{k=1}^KQ_{s_k}^k(t)\left[\hat{r}_k(t)-\sum_{j\in\mathcal{N}_{s_k}}A_{s_kj}^k(\omega(t),{\bm{x}}(t))\right]\nonumber\\&\quad+ \sum_{i\in V}\sum_{k=1}^KQ_i^k(t)\left[\sum_{j:i\in \mathcal{N}_j}{A}_{ji}^k(\omega(t),\bm{x}(t))-\sum_{j\in \mathcal{N}_i}{A}_{ij}^k(\omega(t),\bm{x}(t))\right]
\nonumber\\\le & \sum_{k=1}^KQ_{s_k}^k(\tau)\left[\hat{r}^*_k-\sum_{j\in\mathcal{N}_{s_k}}A_{s_kj}^k(\omega(t),{\bm{x}}(t))\right]\nonumber\\&\quad+ \sum_{i\in V}\sum_{k=1}^KQ_i^k(t)\left[\sum_{j:i\in \mathcal{N}_j}{A}_{ji}^k(\omega(t),\bm{x}(t))-\sum_{j\in \mathcal{N}_i}{A}_{ij}^k(\omega(t),\bm{x}(t))\right]\nonumber\\
&\quad + \sum_{k=1}^K\alpha[(\hat{r}^I_k(t)-\hat{r}^*_k)^2-(\hat{r}^I_k(t+1)-\hat{r}^*_k)^2]+{C} \label{ieq:utilitycross}
\end{align}
By (\ref{ieq:squaredrift}), for the left-hand-side of (\ref{ieq:utilitycross}),
\begin{align}
&\sum_{k=1}^KQ_{s_k}^k(t)\left[\hat{r}_k(t)-\sum_{j\in\mathcal{N}_{s_k}}A_{s_kj}^k(\omega(t),{\bm{x}}(t))\right]\nonumber\\&\quad+ \sum_{i\in V}\sum_{k=1}^KQ_i^k(t)\left[\sum_{j:i\in \mathcal{N}_j}{A}_{ji}^k(\omega(t),\bm{x}(t))-\sum_{j\in \mathcal{N}_i}{A}_{ij}^k(\omega(t),\bm{x}(t))\right]\nonumber\\
\ge & \frac{||\bm{Q}(t+1)||^2-||\bm{Q}(t)||^2}{4} +C. \label{ieq:lhs}
\end{align}
As $(\hat{r}^*_1,\ldots,\hat{r}_K^*)\in Cap(G)$ by definition, for each $\omega$, there exists a set of real numbers $\{a(\omega, \bm{x}),\bm{x}\in\mathcal{X} \}$, $0\le a(\omega, \bm{x})\le 1$ and $\sum_{\bm{x}\in\mathcal{X}}a(\bm{x},\omega)=1$ such that\footnote{Here we assume $\mathcal{X}$ to be discrete. The continuous case follows similarly.}
\begin{align*}
\forall k,\ &\hat{r}_k^*\le \sum_{\omega\in\mathcal{W}}p(\omega)\sum_{j\in\mathcal{N}_{s_k}}\sum_{\bm{x}\in\mathcal{X}}a(\bm{x})A_{s_kj}^k(\omega,\bm{x}), \\
\forall i,k, i\neq s_k,\ &\sum_{\omega\in\mathcal{W}}p(\omega)\sum_{j:i\in \mathcal{N}_j}\sum_{\bm{x}\in\mathcal{X}}a(\omega,\bm{x})A_{ji}^k(\omega, \bm{x})\\&\le \sum_{\omega\in\mathcal{W}}p(\omega)\sum_{j\in \mathcal{N}_i}\sum_{\bm{x}\in\mathcal{X}}a(\omega,\bm{x})A_{ij}^k(\omega, \bm{x}).
\end{align*}
Hence, by Lemma \ref{lemma:maxweight} and follow a similar analysis as (\ref{ieq:slater1}) and (\ref{ieq:slater2}), for the right-hand-side of (\ref{ieq:utilitycross})
\begin{footnotesize}
\begin{align}
&\mathbb{E}\sum_{k=1}^KQ_{s_k}^k(t)\left[\hat{r}^*_k-\sum_{j\in\mathcal{N}_{s_k}}A_{s_kj}^k(\omega(t),{\bm{x}}(t))\right]\nonumber\\
&\ +\mathbb{E}\sum_{i\in V}\sum_{k=1}^KQ_i^k(t)\left[\sum_{j:i\in \mathcal{N}_j}{A}_{ji}^k(\omega(t),\bm{x}(t))-\sum_{j\in \mathcal{N}_i}{A}_{ij}^k(\omega(t),\bm{x}(t))\right]\nonumber\\
\le&\mathbb{E}\sum_{k=1}^KQ_{s_k}^k(t)\left[\hat{r}^*_k-\sum_{\bm{x}\in\mathcal{X}}\sum_{j\in\mathcal{N}_{s_k}}a(\omega(t),\bm{x})A_{s_kj}^k(\omega(t),{\bm{x}}(t))\right]\nonumber\\
&\quad +\mathbb{E}\sum_{i\in V}\sum_{k=1}^KQ_i^k(t)\sum_{\bm{x}\in\mathcal{X}}a(\omega(t),\bm{x})\left[\sum_{j:i\in \mathcal{N}_j}{A}_{ji}^k(\omega(t),\bm{x}(t))-\sum_{j\in \mathcal{N}_i}{A}_{ij}^k(\omega(t),\bm{x}(t))\right]\nonumber\\
=& \mathbb{E}\left[\sum_{k=1}^KQ_{s_k}^k(t)\right]\mathbb{E}\left[\hat{r}^*_k-\sum_{\bm{x}\in\mathcal{X}}\sum_{j\in\mathcal{N}_{s_k}}a(\omega(t),\bm{x})A_{s_kj}^k(\omega(t),{\bm{x}}(t))\right]\nonumber\\
&\ +\mathbb{E}\left[\sum_{i\in V}\sum_{k=1}^KQ_i^k(t)\right]\mathbb{E}\sum_{\bm{x}\in\mathcal{X}}a(\omega(t),\bm{x})\left[\sum_{j:i\in \mathcal{N}_j}{A}_{ji}^k(\omega(t),\bm{x}(t))-\sum_{j\in \mathcal{N}_i}{A}_{ij}^k(\omega(t),\bm{x}(t))\right]\nonumber\\
\le & 0 \label{ieq:rhs}
\end{align}
\end{footnotesize}
Therefore, taking expectation of both sides of (\ref{ieq:utilitycross}) and combining (\ref{ieq:lhs}) and (\ref{ieq:rhs}) yields
\begin{align}
&\mathbb{E}\left[\sum_{k=1}^K V\hat{\nabla}f_k(\hat{r}^I_k(t-1))(\hat{r}^*_k-\hat{r}_k(t-1))\right]   +\frac{\mathbb{E}[||\bm{Q}(t+1)||^2-||\bm{Q}(t)||^2 ]}{4}
\nonumber\\\le & \mathbb{E}\left[\alpha[(\hat{r}^I_k(t-1)-\hat{r}^*_k)^2-(\hat{r}^I_k(t)-\hat{r}^*_k)^2]\right]+C \label{ieq:utility}
\end{align}
By Lemma \ref{lemma:meangrad} and the concavity of $\tilde{f}_k$,
\begin{align*}
 \mathbb{E}\left[\sum_{k=1}^K V\hat{\nabla}f_k(\hat{r}^I_k(t-1))(\hat{r}^*_k-\hat{r}^I_k(t-1))\right] \ge \mathbb{E}[V\sum_{k=1}^K \tilde{f}(\hat{r}_k^*)-\tilde{f}(\hat{r}^I_k(t-1))].
 \end{align*}
  Plugging this in (\ref{ieq:utility}) and rearranging terms, we get
\begin{align}
&\mathbb{E}[V\sum_{k=1}^K \tilde{f}(\hat{r}_k^*)-\tilde{f}(\hat{r}^I_k(t-1))]\nonumber\\
\le & \mathbb{E}\left[\alpha[(\hat{r}^I_k(t-1)-\hat{r}^*_k)^2-(\hat{r}^I_k(t)-\hat{r}^*_k)^2]\right] + \frac{\mathbb{E}[||\bm{Q}(t)||^2-||\bm{Q}(t+1)||^2 ]}{4} +C \label{ieq:telescope}.
\end{align}
Ideally, we would want to sum (\ref{ieq:telescope}) over $t$ and obtain a bound on $\mathbb{E}\left[\sum_{t=1}^{T}\sum_{k=1}^K\tilde{f}_k(\hat{r}^*_k)-\tilde{f}_k(\hat{r}_k(t)) \right]$. However, the parallel instance paradigm brings intricacy to the argument. It stems from the fact that the job-size variables at different time slot may belong to different instances. To make the reasoning clearer, we will write the invoked instance at $t$ at $I_t$ (i.e., $\hat{r}^k(t)=\hat{r}^{I_t}_k(t)$), which makes the dependence explicit but may compromise readability.
First, note that at time $t$, our job-size decisions are $\{r^{I_t}_k(t)\}$, while the left-hand-side of (\ref{ieq:telescope}) is $\mathbb{E}[V\sum_{k=1}^K \tilde{f}(\hat{r}_k^*)-\tilde{f}(\hat{r}^{I_t}_k(t-1))]$. Since the job-size variables of an instance remain unchanged during the intervals when the instance is not invoked, summing the left-hand-side over time $t$, the resulting term differs from $\sum_{t=1}^T\mathbb{E}[V\sum_{k=1}^K \tilde{f}(\hat{r}_k^*)-\tilde{f}(\hat{r}^{I_t}_k(t))]=\sum_{t=1}^T\mathbb{E}[V\sum_{k=1}^K \tilde{f}(\hat{r}_k^*)-\tilde{f}(\hat{r}_k(t))]$ by at most $CV|\mathcal{I}|$, where $|\mathcal{I}|$ is the total number of instances in the reservoir at the end of the time horizon. Second, summing the right-hand-side of (\ref{ieq:telescope}) over time, the term $\frac{\mathbb{E}[||\bm{Q}(t)||^2-||\bm{Q}(t+1)||^2 ]}{4}$ telescopes, but the term $\alpha[(\hat{r}^{I_t}_k(t-1)-\hat{r}^*_k)^2-(\hat{r}^{I_t}_k(t)-\hat{r}^*_k)^2]$ only partially telescopes as the invoked instance $I_t$ may be different for different $t$. More specifically, again due to that the job-size variables of an instance do not change when un-invoked summing the right-hand-side of (\ref{ieq:telescope}) from $t=1$ to $T-1$, we obtain
\begin{align}
&\sum_{I\in\mathcal{I}}\alpha [(\hat{r}^{I}_k(t_I)-\hat{r}^*_k)^2-(\hat{r}^{I}_k(T)-\hat{r}^*_k)^2]+\frac{\mathbb{E}[||\bm{Q}(1)||^2-||\bm{Q}(T)||^2 ]}{4} +CT \nonumber\\
&\le C\alpha |\mathcal{I}| +CT \label{ieq:rhsbound},
\end{align}
where $t_I$ is the time that instance $I$ is created, and (\ref{ieq:rhsbound}) follows from that $||\bm{Q}(1)||$ is bounded by a constant while $||\bm{Q}(T)||^2$ is non-negative.

By the reasoning above, we can see that the key to bound the utility regret is to bound the total number of instance created $|\mathcal{I}|$. By the construction of the parallel-instance paradigm, $|\mathcal{I}|$ is bounded by the maximum delay experienced by the jobs.
We now state a natural assumption that can provide us a handle on the maximum delay through queue lengths.

\begin{assumption}
The network links are work conserving and each job travels through an acyclic route to the destination.
\end{assumption}
The assumption is satisfied by most networks. Under the assumption, using standard queueing-theoretic argument, we have $|\mathcal{I}|\le C\max_{t}\sum_{n\in V, k}Q^k_n(t)$. Using Theorem \ref{thm:workload}, it follows that $|\mathcal{I}|\le \tilde{O}(\sqrt{T})$ with probability at least 1-1/T, which implies that $\mathbb{E}[\mathcal{I}|]\le \tilde{O}(\sqrt{T})$
Therefore, summing (\ref{ieq:telescope}) over time, using (\ref{ieq:rhsbound}) and plugging in the value of $\alpha, V$, we have
\begin{align*}
&\sum_{t=1}^T\mathbb{E}[V\sum_{k=1}^K \tilde{f}(\hat{r}_k^*)-\tilde{f}(\hat{r}_k(t))]\\
\le & C\alpha\mathbb{E}[|\mathcal{I}|] + CT +CV\mathbb{E}[|\mathcal{I}|] \\
\le & \tilde{O}(T).
\end{align*}
Hence, we have 
\begin{align*}
&\sum_{t=1}^T\mathbb{E}[\sum_{k=1}^K \tilde{f}(\hat{r}_k^*)-\tilde{f}(\hat{r}_k(t))]\\
\le &  \tilde{O}(T^{3/4}),
\end{align*}
which demonstrates that the utility regret is of order $\tilde{O}(T^{3/4})$ and finishes the proof of Theorem \ref{thm:regret}.

\subsection{Proof of Auxiliary Results}
\begin{lemma}\label{lemma:aux1}
	Let $\bm{r}^*=(r_1^*,\ldots,r_k^*)$ be the optimal solution to $\mathcal{P}$. Let $\hat{\bm{r}}^*=(\hat{r}_1^*,\ldots,\hat{r}_k^*)$ be the optimal solution to $\mathcal{P}$ restricting to each $r_k\in [\delta,B-\delta]$. $\sum_{k=1}^Kf_k(r_k^*)-f_k(\hat{r}_k^*)\le C\delta$. 
\end{lemma}
\begin{proof}
	Since $\bm{\eta}=(\eta,\ldots,\eta)$ is feasible to $\mathcal{P}$ and $\mathcal{P}$ has convex feasibility region, we have $\tilde{\bm{r}}^*=\frac{\delta}{\eta}\bm{\eta}+(1-\frac{\delta}{\eta})\bm{r}^*$ is feasible to $\mathcal{P}$. Furthermore, observe that for each $k$, $\tilde{r}_k^*\ge \delta$, and by Lipschitz-continuity of $f_k$, $f_k(r^*_k)-f_k(\tilde{r}^*_k)\le C\delta$. Next, define $\bar{\bm{r}}^*$ as $\bar{r}^*_k=\tilde{r}^*_k$ if $\tilde{r}^*_k\le B-\delta$ and $\bar{r}^*_k=B-\delta$ otherwise. Note that for each $k$, $|\bar{r}^*_k - \tilde{r}^*_k|\le \delta$ and $\delta \le \bar{r}^*_k\le B-\delta$. Also, $\bar{\bm{r}}^*$ is feasible to $\mathcal{P}$. Hence, by Lipschitz-continuity of $f_k$, $f_k(\tilde{r}^*_k)-f_k(\bar{r}^*_k)\le C\delta$. Finally, from the definition of $\hat{\bm{r}}^*$, we have $\sum_{k=1}^Kf_k(\hat{r}^*_k)-f_k(\bar{r}^*_k)\ge 0$. Combine the analysis above and the lemma follows.
\end{proof}

\section{Additional Simulation Figures}\label{app:simulationfigure}
In this section, we show figures (Figure \ref{fig:parameter_value1}) of instantaneous utility under the P-GSMW policy with different parameter values.
\begin{figure}[H]
	\subfigure[]{
		\begin{minipage}[]{0.9\linewidth}
			\centering
			\vspace{-2mm}
			\includegraphics[width=1.0\linewidth]{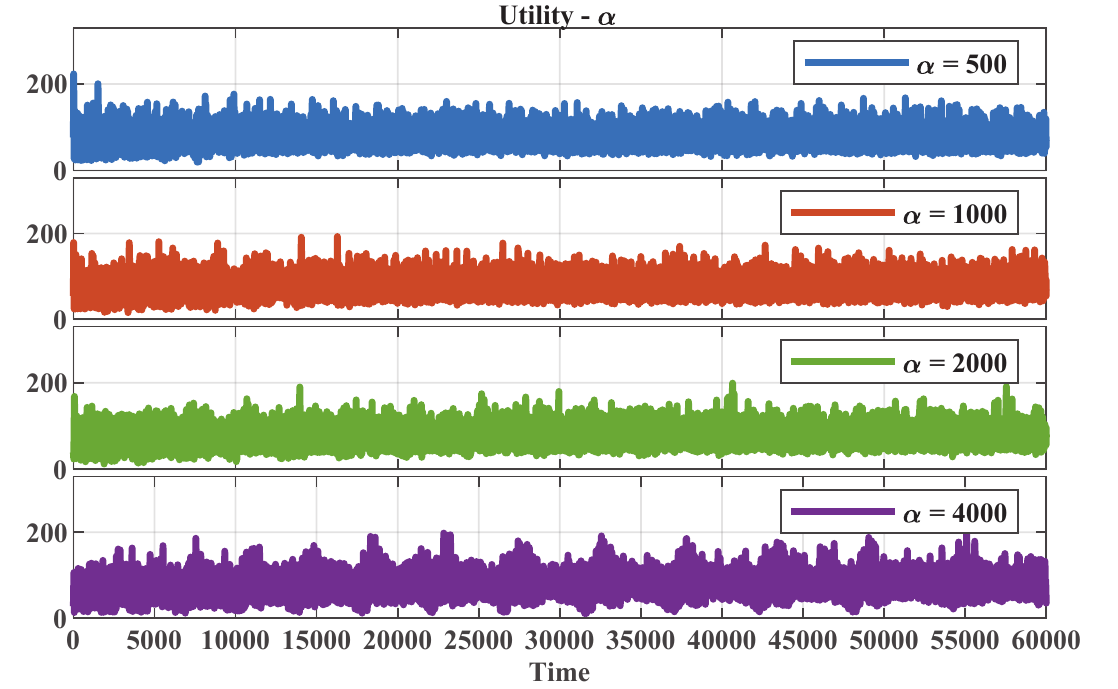}
			\vspace{-3mm}
			\label{fig:utility_alpha_1}
		\end{minipage}%
		
	}
	\subfigure[]{
		\begin{minipage}[]{0.9\linewidth}
			\centering
			\vspace{-2mm}
			\includegraphics[width=1.0\linewidth]{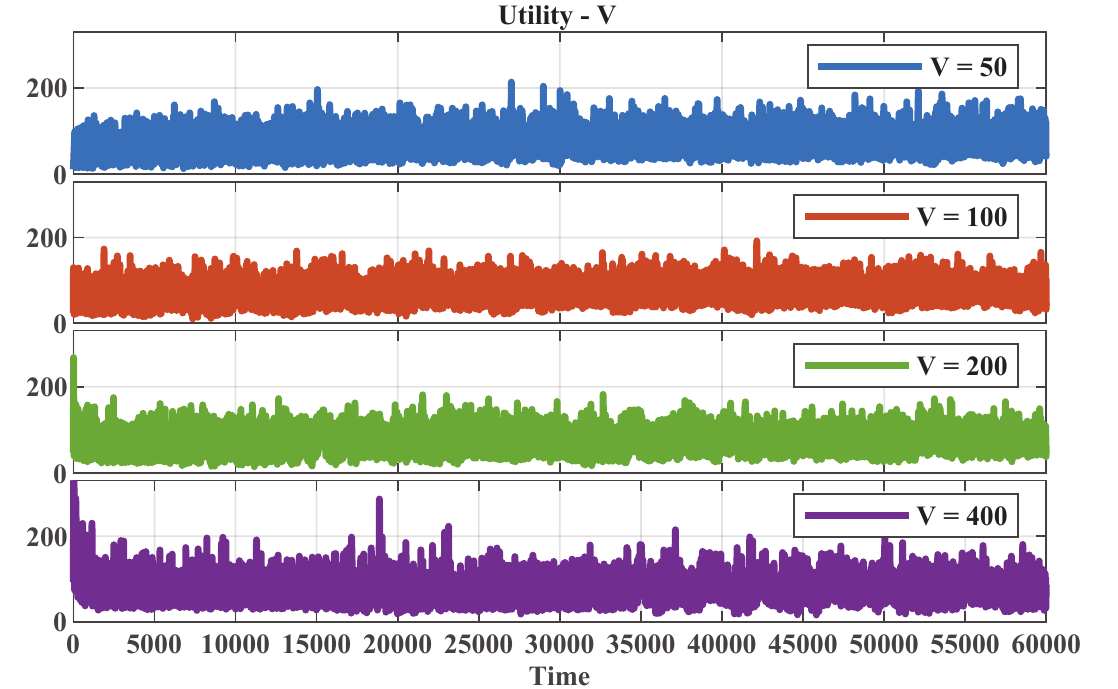}
			\vspace{-3mm}
			\label{fig:utility_V_1}
		\end{minipage}%
		
	}
	\subfigure[]{
		\begin{minipage}[]{0.9\linewidth}
			\centering
			\vspace{-2mm}
			\includegraphics[width=1.0\linewidth]{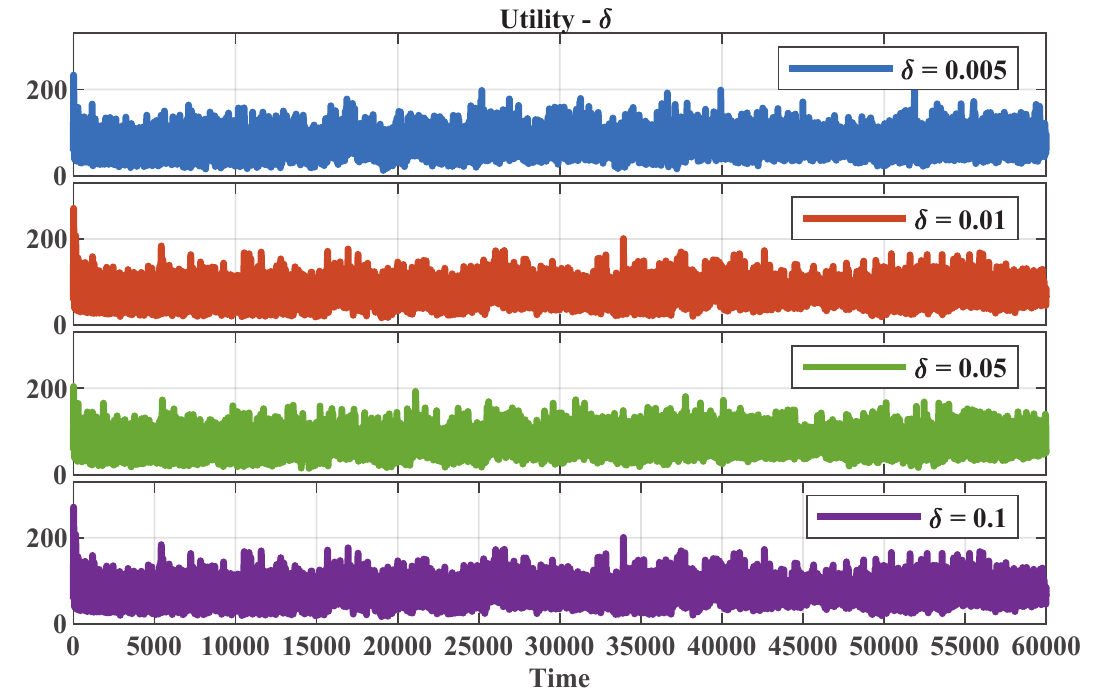}
			\vspace{-3mm}
			\label{fig:utility_delta_1}
		\end{minipage}%
		
	}
	\vspace{-3mm}
	\caption{The instantaneous utility under the P-GSMW policy with different parameter values.}
	\label{fig:parameter_value1}		
\end{figure}

\end{document}